\documentclass[a4paper,11pt,reqno]{article}
\usepackage{relsize}
\usepackage{cite}
\usepackage{color}
\usepackage{mathtools}
\usepackage[dvipsnames]{xcolor}

\usepackage{hyperref, enumitem}

\hypersetup{
  colorlinks   = true, 
  urlcolor     = blue, 
  linkcolor    = Purple, 
  citecolor   = red 
}
\usepackage{amsmath,amsthm,amssymb,mathrsfs}

\usepackage[margin=2.5cm]{geometry}

\usepackage{graphicx}
\usepackage[T1]{fontenc}
\usepackage{authblk}
\usepackage[english]{babel}
\usepackage{stackrel}

\usepackage{array,booktabs,longtable}
\theoremstyle{definition}
\newtheorem{theorem}{Theorem}[section]

\newtheorem{proposition}[theorem]{Proposition}
\newtheorem{corollary}[theorem]{Corollary}
\newtheorem{lemma}[theorem]{Lemma}

\newtheorem{definition}[theorem]{Definition}

\newtheorem{remark}[theorem]{Remark}
\newtheorem{claim}[theorem]{Claim}

\newcommand{\F}{\mathbb F}
\newcommand{\Ftwo}{\mathbb F_2}

\newcommand{\Cm}{\mathcal C_m(q_0)}
\newcommand{\PP}{\mathbb P}

\newcommand{\spn}[1]{\langle #1\rangle_{\Ftwo}}
\newcommand{\spnq}[1]{\langle #1\rangle_{\F_{q_0}}}

\DeclareMathOperator{\tr}{tr}
\DeclareMathOperator{\Tr}{Tr}
\DeclareMathOperator{\wt}{wt}

\DeclareMathOperator{\N}{N}

\title{{\bf Subspace coverings and generalized covering radii of generalized Zetterberg codes}}

\author{Shitao Li\thanks{School of Internet, Anhui University, Hefei, Anhui 230039, China.
E-mail:~ \textsf{lishitao0216@163.com}.},~ 
 Yang Li\thanks{School of Physical and Mathematical Sciences, Nanyang Technological University, Singapore 637371, Singapore. E-mail:~ \textsf{yanglimath@163.com}. },~
 Gaojun Luo\thanks{School of Mathematics, Nanjing University of Aeronautics and Astronautics, Nanjing, Jiangsu 211106, China. E-mail:~ \textsf{gaojun\_luo@nuaa.edu.cn}.},~
 Zhonghua Sun\thanks{School of Mathematics, Hefei University of Technology, Hefei, Anhui, 230601, China. E-mail:~ \textsf{sunzhonghuas@163.com}.}
 }

\begin{document}
\maketitle

\begin{abstract}
Generalized covering radii measure how many columns of a parity-check matrix are needed to generate several syndromes simultaneously. Their finite-geometric counterparts are $(\rho,t)$-saturating sets, for which every $t$-dimensional subspace is contained in a subspace generated by at most $\rho$ prescribed vectors. We investigate this covering problem for the norm-one configurations associated with generalized Zetterberg codes. We establish the upper bound $2t+1$ over every nonbinary finite field and in an explicit binary range, together with complementary lower bounds obtained by counting subspaces and constructing subfield obstructions. For an explicit range of large $t$, these configurations are $t$-strong blocking sets, and the $t^{\rm th}$ generalized covering radius attains its minimum possible value $t$. For binary Zetterberg codes, we determine the second generalized covering radius in every extension degree and prove that the third radius is seven for an infinite subfamily.
\end{abstract}



\section{Introduction}
Covering problems form a classical theme in coding theory; see, for instance, \cite{CHLL1997}.  Given a linear code \(C\subseteq \F_{q_0}^{n}\), its covering radius \(\rho(C)\) is the smallest integer \(\rho\) such that every vector of \(\F_{q_0}^{n}\) lies at Hamming distance at most \(\rho\) from some codeword of \(C\). Equivalently, if \(H\) is a parity-check matrix of \(C\), then \(\rho(C)\) is the smallest integer \(\rho\) such that every syndrome belongs to the span of at most \(\rho\) columns of \(H\).
This syndrome formulation provides a natural bridge to finite geometry:~ the columns of \(H\) may be regarded as vectors in a finite-dimensional space, such that covering properties of the code are translated into spanning properties of the corresponding finite configuration \cite{D1995,DFMP2005,DMP2003}.
Related projective spanning conditions also arise in the theory of strong and cutting blocking sets, which have been studied in close connection with linear codes; see \cite{ABNR2022,ABDN2024,BT2026}.
This viewpoint is closely related to the theory of saturating sets, complete caps, and other covering configurations in finite projective spaces \cite{CCMP2023,G2013,HS2001}.

The covering radius is also an important parameter from the coding-theoretic point of view. It arises naturally in decoding, data compression, testing, write-once memories, and related combinatorial covering problems (see \cite{BLP1998,CHLL1997,CKMS1985} and the references therein).
Moreover, it is closely connected with several classical extremal classes of codes, most notably perfect and quasi-perfect codes \cite{CKMS1985,CHLL1997,HuffmanPless2003}.
However, determining the covering radius is difficult even for structured codes \cite{DJ-MC}. It has been proven in \cite{M-IT1984} that computing the exact covering radius of a code is both NP-hard and co-NP-hard. Therefore,  much of the literature concerns upper and lower bounds rather than exact values \cite{AB2002,T-DAM1987,JCTA1}.
Exact covering radii are known only for selected structured families, including Melas, BCH, cyclic, and Reed-Muller codes. See, for example, \cite{SHOS2022,MC-IT2003,GPZ1960,Helleseth1985,GaoEtAl2023}.

Zetterberg codes form one of the classical families for which the covering problem has particularly rich algebraic and arithmetic structure. The binary Zetterberg code is a cyclic code whose parity-check columns are given by the norm-one subgroup of a quadratic finite-field extension \cite{Zetterberg1962}.
Dodunekov \cite{Dodunekov1985} proved that the double-error-correcting binary Zetterberg codes of length $2^{2s}+1$, with $s\ge2$, have covering radius three and are quasi-perfect. The related ternary quasi-perfect codes were constructed by Gashkov and Sidel'nikov~\cite{GS1986}.
More recently, generalized Zetterberg codes over arbitrary finite fields and their half-Zetterberg variants have been studied systematically in \cite{SHO2023,SLHO2025,SHO2026,XiongYan2026}.
In particular, Xiong and Yan \cite{XiongYan2026} established a uniform upper bound of three for the ordinary covering radius over every finite field by means of character sums and Weil-type estimates, and determined the exact radius in broad parameter ranges.

Motivated by the study of database linear querying, Elimelech, Firer, and Schwartz~\cite{EFS2021} introduced the concept of \emph{generalized covering radii} by replacing a single syndrome by a collection of syndromes that must share one set of columns. More precisely, if \(H\in \F_{q_0}^{k\times n}\) is a parity-check matrix of \(C\), then \(\rho_t(C)\) is the smallest integer \(r\) such that every \(t\)-dimensional subspace of \(\mathbb F_{q_0}^{k}\) is contained in the span of at most \(r\) columns of \(H\).
The case \(t=1\) recovers the ordinary covering radius. Equivalently, one may formulate the parameter in terms of the union of the supports of representatives of several cosets, or as an ordinary covering radius after extension of scalars \cite{EFS2021}.
The generalized covering radii of several classical families have recently been investigated.
Bounds and exact values for Reed-Muller codes were obtained in \cite{EWW2022}, while Elimelech and Schwartz \cite{ES2024} studied the second-order football-pool problem and determined the corresponding asymptotic rate.
Exact second generalized covering radii for binary BCH and related cyclic codes were obtained in \cite{YS2025,LuoZhouEtAl2026}, and generalized covering radii of double-error-correcting BCH codes were studied in \cite{OO2026,XiongYip2026}.
For Melas codes, Luo {\em et al.} \cite{LYYCO2026} obtained bounds in the binary case and determined the third radius to be seven in a specified range. Li and Xiong~\cite{LiXiong2026} determined the second radii over all finite fields and established analogous higher-order bounds.
It was also noted that there are many gaps in the above results. Recently, a geometric framework for these parameters was developed by Alfarano, Marino, Neri, and Trombetti~\cite{AMNT2026}, linking generalized covering radii with subspace-spanning configurations in finite projective geometry.

In this paper, we focus on generalized covering radii of generalized Zetterberg codes. Let $q_0$ be a prime power, let $m\ge1$, and put $q=q_0^m$. Let \(\Cm\) denote the \(q_0\)-ary generalized Zetterberg code of length \(q_0^m+1\) and redundancy $2m$. 
Its definition is recalled in Section~\ref{subsec1}. Set $\delta=\gcd(2,q_0-1)$, the number of norm-one representatives of each projective column direction. Our first result gives upper and lower bounds for the whole family.

\begin{theorem}\label{thm:~main-bounds}
Let $1\le t\le2m$. The following results hold. 
\begin{enumerate}
\item [(i)] If $q_0\ge3$, then
\begin{equation*}
 \rho_t(\Cm)\le\min\{2t+1,2m\}.
\end{equation*}
For $q_0=2$, the same bound holds whenever $m\ge2t+2$.

\item [(ii)] For every $q_0$, if $2\le t\le m$ and \( q\ge\Theta_{q_0,t}:=\frac{q_0^{t(2t-1)}}{\delta^{2t-1}(2t-1)!}\),
then $\rho_t(\Cm)\ge2t$.

\item [(iii)] Suppose $m=2s$. Then
\begin{equation}\label{eq:~h-lower}
 \rho_t(\mathcal C_{2s}(q_0))\ge h_{q_0}(t):=\min\{\ell\ge t:~q_0^{\ell-t}\ge1+(q_0-1)\ell\}
\end{equation}
for $1\le t\le s-1$ if $q_0$ is even, and for $1\le t\le s$
if $q_0$ is odd.
\end{enumerate}
\end{theorem}

Alfarano, Marino, Neri, and Trombetti~\cite{AMNT2026} showed that the extremal case \(\rho=t\) is equivalent to the strong blocking property:~ every \(t\)-dimensional subspace is generated by the column directions that it already contains.
The following theorem shows the generalized covering radius of generalized Zetterberg codes attains the dimension lower bound for large $t$.

\begin{theorem}\label{thm:~main-high}
If $1\le t\le2m$ and $\Delta_t:=q+1-2(2q_0^{2m-t}-1)\sqrt q>0$, then 
\begin{equation*}
 \rho_t(\Cm)=t.
\end{equation*}
In particular, for $s\ge2$ this holds for $\mathcal C_{2s}(q_0)$ throughout $3s+2\le t\le4s$ for every $q_0$, and throughout $3s+1\le t\le4s$ when $q_0\ge4$.
\end{theorem}

Finally let $\mathcal B_m=\mathcal C_m(2)$ and $\mathcal Z_s=\mathcal B_{2s}$ for $s\ge 2$.
The binary configuration supplies additional exact values that do not follow from the general bounds alone.
\begin{theorem}\label{thm:~binary-main}
The following results hold.  
\begin{enumerate}
\item [(i)] For every $m\ge1$,
\begin{equation*}
 \rho_2(\mathcal B_m)=
 \begin{cases}2,&m=1,\\3,&m=2,\\6,&m=3,\\5,&m\ge4.\end{cases}
\end{equation*}

\item [(ii)] If $8\mid m$, then
$\rho_3(\mathcal B_m)=7$. Equivalently, $\rho_3(\mathcal Z_s)=7$
when $4\mid s$.
\end{enumerate}
\end{theorem}

Theorem~\ref{thm:~binary-main}(i) has a theoretical proof in every degree. Part~(ii) uses explicit polynomial certificates over $\F_{256}$ together with a uniform extension-degree argument. Table~\ref{tab:~all-results} records the bounds and exact values of covering radius, with their parameter restrictions. All orders are between one and the redundancy unless stated otherwise.

\begingroup
\small\setlength{\tabcolsep}{4pt}\renewcommand{\arraystretch}{1.55}
\setlength{\LTpre}{.5\baselineskip}\setlength{\LTpost}{.5\baselineskip}
\begin{longtable}{@{}
>{\raggedright\arraybackslash}p{16mm}
>{\raggedright\arraybackslash}p{60mm}
>{\raggedright\arraybackslash}p{55mm}
>{\raggedright\arraybackslash}p{20mm}@{}}
\caption{Generalized covering radii of generalized Zetterberg codes}\label{tab:~all-results}\\
\toprule Code & Condition & Generalized covering radius & Ref. \\\midrule
\endfirsthead
\bottomrule\endlastfoot
$\mathcal C_m(q_0)$ & $q_0\ge3$ & $\rho_t\le\min\{2t+1,2m\}$ & Thm.~\ref{thm:~main-bounds}(i)\\
$\mathcal B_m$ & $m\ge2t+2$ & $\rho_t\le2t+1$ & Thm.~\ref{thm:~main-bounds}(i)\\
$\mathcal C_m(q_0)$ & $2\le t\le m$, $q\ge\Theta_{q_0,t}$ & $\rho_t\ge2t$; $\rho_t\in\{2t,2t+1\}$ if $q_0\ge3$ & Thm.~\ref{thm:~main-bounds}(ii)\\ 
$\mathcal C_{2s}(q_0)$ & $t\le s-1$ if $q_0$ is even; $t\le s$ if $q_0$ is odd & $\rho_t\ge h_{q_0}(t)$ & Thm.~\ref{thm:~main-bounds}(iii)\\ \hline 
$\mathcal C_m(q_0)$ & $\rho_1=2$, $2\le t\le m$, $q\ge\Theta_{q_0,t}$ & $\rho_t=2t$ & Cor.~\ref{cor:~exact-2t}\\
$\mathcal C_m(q_0)$ & $\Delta_t>0$ & $\rho_t=t$ & Thm.~\ref{thm:~main-high}\\
$\mathcal C_{2s}(q_0)$ & $s\ge2$; $t\ge3s+2$, or $q_0\ge4$ and $t\ge3s+1$ & $\rho_t=t$ for $t\le4s$ & Thm.~\ref{thm:~main-high}\\
$\mathcal C_{2s+1}(q_0)$ & $s\ge1$; $t\ge3s+2$ if $q_0\ge16$; $t\ge3s+3$ if $3\le q_0<16$; $t\ge3s+4$ if $q_0=2$ & $\rho_t=t$ for $t\le4s+2$ & Cor.~\ref{cor:~odd-high}\\
$\mathcal B_m$ & $m=1$; $m=2$; $m=3$; $m\ge4$ & $\rho_2=2$; $3$; $6$; $5$ & Thm.~\ref{thm:~binary-main}(i)\\
$\mathcal Z_s$ &  $4\mid s$ & $\rho_3=7$ & Thm.~\ref{thm:~binary-main}(ii)\\
\end{longtable}
\endgroup

The rest of the paper is organized as follows. In Section~\ref{sec:~preliminaries}, we recall the necessary background on generalized covering radius and generalized Zetterberg codes. In Section~\ref{sec:~bounds}, we proves the general upper and lower bounds in Theorem \ref{thm:~main-bounds}. In Section~\ref{sec:~high}, we give a proof of Therem \ref{thm:~main-high}. In Section~\ref{sec:~binary}, we use additional binary structure to determine the exact second radii and the third radius on an infinite subfamily. In Section~\ref{sec:~conclusion}, we conclude this paper.

\section{Preliminaries}\label{sec:~preliminaries}
In this section we provide the necessary background material for the rest of the paper. Throughout the paper, let $q_0$ be a prime power and $m$ a positive integer. Put
\(q=q_0^m,\ \F_{q_0}\subseteq\F_q\subseteq\F_{q^2},\) and the norm-one group \(U=\{u\in\F_{q^2}^*:~u^{q+1}=1\}.\)
All dimensions and linear spans are over $\F_{q_0}$ unless another field is specified.
In particular, a $t$-dimensional subspace corresponds to a projective $(t-1)$-space.
We write $[n]=\{1,\ldots,n\}$ and $\F_a^*=\F_a\setminus\{0\}$. 

\subsection{Linear codes and generalized Zetterberg codes}\label{subsec1}

We first recall some standard notation from coding theory. An \([n,n-k]_{q_0}\) linear code is an \((n-k)\)-dimensional \(\F_{q_0}\)-subspace \(C\le \F_{q_0}^{n}\). The elements of \(C\) are called \emph{codewords}. The support of a vector is the set of its nonzero coordinates, its Hamming weight is the size of this support, and the Hamming distance between two vectors is the weight of their difference. The minimum distance of \(C\) is \(d(C)=\min\{\wt_H(c):~c\in C,\ c\ne0\}.\) When \(d(C)=d\), we write \(C\) as an \([n,n-k,d]_{q_0}\) code. We write $[n,k,\ge d]_{q_0}$ when only a lower bound on the minimum distance is given.
A full-rank matrix \(H\in\F_{q_0}^{k\times n}\) is called a \emph{parity-check matrix} of \(C\) if \(C=\{x\in\F_{q_0}^{n}:~Hx^T=0\}.\)
Let \(H=(\mathbf{h}_1,\ldots,\mathbf{h}_n)\), where \(\mathbf{h}_i\in\F_{q_0}^{k}\). The \emph{syndrome} of a vector
\(x=(x_1,\ldots,x_n)\in\F_{q_0}^{n}\) is
\[Hx^T=\sum_{i=1}^n x_i \mathbf{h}_i\in\F_{q_0}^{k}.\]
Thus, \(x\in C\) if and only if \(Hx^T=0\). The space \(\F_{q_0}^{k}\) will be referred to as the \emph{syndrome space} of \(C\).
In particular, a syndrome \(s\in\F_{q_0}^{k}\) can be represented using a set of coordinates \(I\subseteq[n]\) precisely when
\(s\in\left\langle \mathbf{h}_i:~i\in I\right\rangle_{\F_{q_0}}.\)
This interpretation is the form used throughout the paper and leads directly to the definition of generalized covering radii in the next subsection.

The multiplicative group $U$ is cyclic group of order $q+1$. Fix a generator $\alpha$. The generalized Zetterberg code over $\F_{q_0}$ is
\begin{equation*}
 \Cm=\left\{(c_u)_{u\in U}\in\F_{q_0}^{\,q+1}:~\sum_{u\in U}c_u u=0\right\}.
\end{equation*}
It is a $[q+1,q+1-2m]_{q_0}$ cyclic code generated by the minimal polynomial of $\alpha$ over $\F_{q_0}$ \cite{SHO2023,XiongYan2026}. 
More precisely, choose an $\F_{q_0}$-basis of $\F_{q^2}$ and let $\phi:~\F_{q^2}\to\F_{q_0}^{\,2m}$ be the associated coordinate isomorphism. Then
\[H=\bigl(\phi(1)\ \phi(\alpha)\ \cdots\ \phi(\alpha^q)\bigr)\]
is a parity-check matrix over $\F_{q_0}$. We identify its columns with the corresponding elements of $U$, but all column combinations have coefficients in $\F_{q_0}$.

The binary specialization will be denoted by \(\mathcal B_m=\mathcal C_m(2).\)
For $s\ge2$, the even-degree subfamily $\mathcal Z_s=\mathcal B_{2s}$ is the classical double-error-correcting Zetterberg family. It has parameters $[2^{2s}+1,2^{2s}+1-4s,5]_2$ and covering radius three, and is therefore quasi-perfect~\cite{Zetterberg1962,Dodunekov1985}. More generally, $\mathcal B_m$ has minimum distance three for odd $m$ and five for even $m$. Its ordinary covering radius is one for $m=1$, two for $m=2$, and three for $m\ge3$~\cite{XiongYan2026}.

Two units $u,v\in U$ represent the same projective column direction if and only if $u/v\in\F_{q_0}^*$ and $(u/v)^2=1$. Hence each direction has $\delta=\gcd(2,q_0-1)$ representatives in $U$, and there are $n_0=(q+1)/\delta$ distinct directions. Fix a representative set $U_0\subseteq U$, taking $U_0=U$ when $q_0$ is even and $U_0=\{1,\alpha,\ldots,\alpha^{n_0-1}\}$ when $q_0$ is odd.
In the latter case $U=U_0\sqcup(-U_0)$. The coordinate vectors of $U_0$ form a matrix $H_0$ of  rank $2m$, called the reduced parity-check matrix here. These representatives are used in column counts and shortening arguments.

\subsection{Generalized covering radii}\label{subsec:~covering-definitions}

We recall the definition of generalized covering radii introduced
in~\cite{EFS2021}.

\begin{definition}
Let \(C\) be an \([n,n-k]_{q_0}\) linear code with a parity-check matrix \(H=(h_1,\ldots,h_n)\in\F_{q_0}^{\,k\times n}\). For \(1\le t\le k\), the $t^{\rm th}$ generalized covering radius of \(C\) is
\begin{equation}\label{eq:~general-rho}
 \rho_t(C)= \max_{s_1,\ldots,s_t\in\F_{q_0}^{k}} \min\left\{ |I|:~ I\subseteq[n],
 ~s_1,\ldots,s_t\in \left\langle h_i:~i\in I\right\rangle_{\F_{q_0}} \right\}.
\end{equation}
\end{definition}

Thus, \(\rho_t(C)\) is the smallest integer \(r\) such that any \(t\) syndromes can be generated simultaneously by a common set of at most \(r\) columns of \(H\). The coefficients used to represent the individual syndromes may be different. This parameter is independent of the choice of a parity-check matrix, and \(\rho_1(C)\) is the ordinary covering radius. 
There are also several equivalent definitions of generalized covering radii for linear codes in \cite{EFS2021}. We shall repeatedly use the standard properties \cite{EFS2021}
\begin{equation}\label{eq:~basic}
 t\le\rho_t(C)\le k,~~
 \rho_t(C)\le\rho_{t+1}(C),~~
 \rho_{a+b}(C)\le\rho_a(C)+\rho_b(C),
\end{equation}
whenever the relevant indices do not exceed \(k\).

Recently, a geometric framework for generalized covering radius was developed in~\cite{AMNT2026}. They proved that it is convenient to work directly with subspaces since a set of \(t\) syndromes spans a subspace of dimension at most \(t\).
A \(t\)-tuple of syndromes spans a subspace of dimension at most \(t\), which can be extended to a \(t\)-dimensional subspace. Conversely, a basis of any \(t\)-dimensional subspace is an admissible \(t\)-tuple in~\eqref{eq:~general-rho}. Hence
\[ \rho_t(C)= \max_{\substack{S\le\F_{q_0}^{k}\\ \dim S=t}}\min\left\{|I|:~I\subseteq[n],
\ S\subseteq\left\langle h_i:~i\in I\right\rangle_{\F_{q_0}} \right\}.\]
This viewpoint is closely related to the geometric framework of \((r,t)\)-saturating sets developed in~\cite{AMNT2026}.
Let \(\mathcal H=\{h_1,\ldots,h_n\}\) denote the column configuration of \(H\). Since a \(t\)-dimensional subspace cannot be spanned by fewer than \(t\) columns, the minimum possible value of the $t^{\rm th}$ generalized covering radius is \(t\). Moreover,
\begin{equation}\label{eq:~internal-spanning}
 \rho_t(C)=t~~\Longleftrightarrow~~
 \left\langle\mathcal H\cap S\right\rangle_{\F_{q_0}}=S
 ~\text{for every \(t\)-dimensional subspace } S\le\F_{q_0}^{k}.
\end{equation}
Indeed, if \(S\) is covered by exactly \(t\) columns, then those columns form a basis of \(S\) and hence lie in \(S\). Conversely, if the columns contained in \(S\) span \(S\), then \(S\) contains a column basis of size \(t\). In the terminology of~\cite{AMNT2026}, \(\mathcal H\) is \((r,t)\)-saturating when \(\rho_t(C)\le r\), while the condition on the right-hand side of \eqref{eq:~internal-spanning} is the \(t\)-strong blocking property.

For the generalized Zetterberg code, we identify the syndrome space with the \(2m\)-dimensional \(\F_{q_0}\)-space \(\F_{q^2}\), and the columns of the parity-check matrix with the elements of \(U\). For an \(\F_{q_0}\)-subspace \(W\le\F_{q^2}\), define
\[ \tau_U(W)= \min\left\{ |I|:~ I\subseteq U,\; W\subseteq\langle I\rangle_{\F_{q_0}} \right\}.\]
Then
\begin{equation}\label{eq-gcv3}
 \rho_t(\mathcal C_m(q_0))=\rho_t(U) := \max_{\substack{W\le\F_{q^2}\\ \dim W=t}} \tau_U(W).
\end{equation}
In particular, \(\rho_t(\mathcal C_m(q_0))=t\) in and only if \(\spnq{U\cap W}=W\) for every \(t\)-dimensional \(W\le\F_{q^2}\).

A nonzero syndrome \(a\in\F_{q^2}\) is called \emph{deep} if \(\tau_U(\F_{q_0}a)=\rho_1(U)\). A subspace is called a \emph{deep-syndrome subspace} if all of its nonzero elements are deep. In particular, when \(\rho_1(U)=3\), a nonzero syndrome is deep precisely when it cannot be represented using one or two columns of \(U\).
The formulation~\eqref{eq-gcv3} will be used throughout the paper. To prove an upper bound \(\rho_t(\mathcal C_m(q_0))\le r\), it is enough to show that \(\tau_U(W)\le r\) for every \(t\)-dimensional subspace \(W\le\F_{q^2}\). To prove a lower bound \(\rho_t(\mathcal C_m(q_0))\ge r\), it is enough to exhibit one \(t\)-dimensional subspace \(W\) with \(\tau_U(W)\ge r\).
Consequently, an exact value \(\rho_t(\mathcal C_m(q_0))=r\) follows by establishing these two bounds simultaneously.

\subsection{Rational functions and degree-two covers}
We recall a few standard facts about rational functions and degree-two covers that will be used later in converting character sums into point counts on algebraic curves. Whenever a curve is given by an affine equation, \(C\) denotes its smooth projective model.

Let \(K\) be a field and let \(g\in K(T)\) be nonzero. For \(P\in\PP^1(\overline K)\), choose a local parameter \(t_P\) at \(P\). Then \(g=t_P^m u\) and \(u(P)\ne0\) for a unique integer \(m\). This integer is the order of \(g\) at \(P\), denoted by \({\rm ord}_P(g)\). Thus, \({\rm ord}_P(g)=d>0\) means that \(g\) has a zero of order \(d\) at \(P\), while \({\rm ord}_P(g)=-d<0\) means that \(g\) has a pole of order \(d\).
At infinity, one may use \(S=1/T\) as a local parameter. Hence
\({\rm ord}_\infty(g) ={\rm ord}_{S=0}\bigl(g(1/S)\bigr).\)
In particular, if \(g(T)=\frac{P(T)}{Q(T)}\) with coprime polynomials \(P,Q\in K[T]\), then
\({\rm ord}_\infty(g)=\deg Q-\deg P.\)
We shall also use that an invertible M\"obius change of variable preserves zero and pole orders. More precisely, if
\[\phi(T)=\frac{aT+b}{cT+d}, \quad ad-bc\ne0,\]
then
\({\rm ord}_P(g\circ\phi) = {\rm ord}_{\phi(P)}(g).\)

We next recall the two types of degree-two covers that occur below. Suppose first that \(\operatorname{char}K=2\), and consider the Artin-Schreier curve
\begin{equation*}
 C:~ Y^2+Y=g(T).
\end{equation*}
Replacing \(g\) by \(g+h^2+h\) does not change the corresponding Artin-Schreier extension. Thus, \(g\) may be taken in reduced form, in which all pole orders are odd. If these pole orders are \(d_1,\ldots,d_\ell\), then by the Artin-Schreier genus formula \cite[Lemma~2.6]{PriesZhu2012}, we have
\begin{equation}\label{eq:~AS-genus}
 g(C) = \frac{\sum_{i=1}^{\ell}(d_i+1)-2}{2}.
\end{equation}
In particular, two simple poles give \(g(C)=1\).
Moreover, the presence of a simple pole shows that \(g\) cannot be of the form \(h^2+h\), and hence  the corresponding Artin-Schreier cover is geometrically irreducible.

Now suppose that \(\operatorname{char}K\ne2\), and consider
\begin{equation*}
 C:~ Y^2=g(T).
\end{equation*}
Assume that \(g\) is not a square in \(\overline K(T)\), so that the cover is geometrically irreducible. A point \(P\in\PP^1(\overline K)\) is a branch point precisely when \({\rm ord}_P(g)\) is odd. The point of \(C\) lying above such a branch point is ramified with ramification index \(2\). Thus, simple zeros and simple poles give branch points, whereas zeros and poles of even order are unramified.
For the degree-two map \(\pi:~C\to\PP^1\), it follows from the Riemann-Hurwitz formula \cite[Chapter~III]{Stichtenoth2009} that
\begin{equation}\label{eq:~RH-double}
 2g(C)-2 = 2\bigl(2g(\PP^1)-2\bigr) +\sum_{Q\in C}(e_Q-1),
\end{equation}
where \(e_Q\) is the ramification index at \(Q\). Since \(g(\PP^1)=0\), each ramified point contributes one to the last sum. Hence a quadratic cover with four branch points has genus one.

Finally, if \(C/\F_q\) is a smooth projective geometrically irreducible curve of genus \(g\), then by the Hasse-Weil bound~\cite{Poonen2006}, we have
$
 \bigl|\#C(\F_q)-(q+1)\bigr|\le2g\sqrt q.
$
In particular, for a genus-one curve,
\begin{equation}\label{eq:~HW-genus-one}
 \bigl|\#C(\F_q)-(q+1)\bigr| \le2\sqrt q.
\end{equation}

\section{Upper and lower bounds on the generalized covering radii of generalized Zetterberg codes}\label{sec:~bounds}

In this section, we prove Theorem~\ref{thm:~main-bounds}. The two lower bounds come from different sources. One is obtained by counting the subspaces generated by small sets of columns and comparing them with all possible target subspaces. The other is based on constructing subspaces whose nonzero vectors admit no sufficiently short column representation. For the upper bound, we first identify a common summand that works for many vectors and then choose a basis adapted to that summand.
We begin by fixing some notation used throughout the section. Let
\[ \N(x)=\N_{\F_{q^2}/\F_q}(x)=x^{q+1}\]
denote the norm from \(\F_{q^2}\) to \(\F_q\). In particular, \(\N(a)=a^2\) for \(a\in\F_q\), while \(u^q=u^{-1}\) for \(u\in U\). The nonzero syndromes representable by a single column are exactly
\begin{equation*}
 \F_{q_0}^*U = \left\{a\in\F_{q^2}^*:~\N(a)\in(\F_{q_0}^*)^2 \right\}.
\end{equation*}
Indeed, \(\N(\lambda u)=\lambda^2\) for \(\lambda\in\F_{q_0}^*\), and conversely \(\N(a)=\lambda^2\) implies \(a/\lambda\in U\).

For nested finite fields, we write \(\Tr_{\F_b/\F_a}\) for the field trace. When \(q\) is even, we abbreviate \(\tr=\Tr_{\F_q/\F_2}\). In particular, \(\tr(z^2)=\tr(z)\). When \(q\) is odd, \(\chi\) denotes the quadratic character of \(\F_q\), extended by \(\chi(0)=0\).
In characteristic \(p\), the canonical additive character of \(\F_{p^e}\) is
\[ z\mapsto \exp\!\left(\frac{2\pi i}{p} \Tr_{\F_{p^e}/\F_p}(z)\right).\]
All traces and characters appearing in the representation counts below are taken on \(\F_q\), rather than on the syndrome space \(\F_{q^2}\).

\subsection{Counting and subfield lower bounds}
The counting argument is most effective when the extension degree is large relative to \(t\), whereas the subfield construction below does not require such a size condition. We first prove the counting bound. The underlying ordered-tuple argument is standard in covering problems; see, for example, \cite[Theorem~IV.1]{XiongYip2026} and \cite[Theorem~6]{LiXiong2026}. Here we apply it directly to the distinct projective columns of the Zetterberg configuration.

\begin{proof}[Proof of Theorem~\ref{thm:~main-bounds}\textup{(ii)}]
Let $2\le t\le m$ and suppose, to the contrary, that $\rho_t(\Cm)\le r$, where $r=2t-1$. Recall that $U_0\subseteq U$ contains one representative of each projective column direction. Hence
\(|U_0|=n_0=\frac{q+1}{\delta},\) where \(\delta=\gcd(2,q_0-1).\) Replacing a column by a nonzero $\F_{q_0}$-multiple and removing repeated directions preserve its span. Thus, the covering assumption also holds for $U_0$. Moreover, $U_0$ spans the $2m$-dimensional $\F_{q_0}$-space $\F_{q^2}$, and hence $n_0\ge2m>r\ge3$.

For each \(r\)-subset \(I\subseteq U_0\), let \(V_I=\spnq{I}\). By the definition of the generalized covering radius, every ordered \(t\)-tuple \((s_1,\ldots,s_t)\in(\F_{q^2})^t\) is contained in the span of a common set of at most \(r\) columns from \(U_0\). Since \(r<n_0\), any smaller covering set may be enlarged to an \(r\)-subset. Consequently,
\[(\F_{q^2})^t =\bigcup_{\substack{I\subseteq U_0\\|I|=r}}V_I^t,\]
where $V_I^t$ denotes the $t$-fold Cartesian product of $V_I$.
The tuples here need not consist of distinct or linearly independent syndromes. There are exactly $q^{2t}$ such tuples, whereas a fixed $I$ covers \(|V_I^t|=|V_I|^{t}\le q_0^{rt}\) tuples. Taking cardinalities and allowing overlaps among the sets $V_I^t$, we obtain
\begin{equation}\label{eq:~count-necessary}
 q^{2t} \le \sum_{\substack{I\subseteq U_0\\|I|=r}}q_0^{t\dim_{\F_{q_0}}V_I}
 \le \binom{n_0}{r}q_0^{rt}.
\end{equation}

It remains to estimate the binomial coefficient. Since \(n_0=(q+1)/\delta\), 
\[\binom{n_0}{r}=\frac{1}{\delta^r r!}\prod_{j=0}^{r-1}(q+1-\delta j).\]
Because \(r\ge3\) and \(\delta\ge1\), the first three factors are at most \(q+1\), \(q\), and \(q-1\), respectively, while every remaining factor is at most \(q\). Hence
\[ \binom{n_0}{r} \le \frac{(q+1)q(q-1)q^{r-3}}{\delta^r r!} < \frac{q^r}{\delta^r r!},\]
where the strict inequality follows from $(q+1)q(q-1)=q^3-q<q^3$.
Substituting this estimate into~\eqref{eq:~count-necessary} and using
\(r=2t-1\), we get
\[ q^{2t} < q^{2t-1}\frac{q_0^{t(2t-1)}}{\delta^{2t-1}(2t-1)!} = q^{2t-1}\Theta_{q_0,t}.\]
Thus, $q<\Theta_{q_0,t}$, contradicting the assumption $q\ge\Theta_{q_0,t}$. It follows that $\rho_t(\Cm)\ge2t$.
\end{proof}

We next derive the second lower bound by constructing subspaces whose nonzero vectors admit no representation using at most two columns. A common cover of such a subspace naturally gives rise to a short linear code of minimum distance at least three, to which the Hamming bound can be applied.

\begin{proof}[Proof of Theorem~\ref{thm:~main-bounds}\textup{(iii)}]
Let \(m=2s\). Then \(q=q_0^{2s}\), and \([\F_q:\F_{q_0^s}]=2\). By the syndrome-subspace formulation of the generalized covering radius in \eqref{eq-gcv3}, it is enough to construct a \(t\)-dimensional subspace \(W\le\F_q\) such that \(\tau_U(W)\ge h_{q_0}(t)\).

We first isolate the obstruction to one- and two-column
representations.

\begin{claim}\label{claim1}
Let \(a\in\F_q\setminus\F_{q_0}\) satisfy
\(a^2\in\F_{q_0^s}\). Then \(a\) cannot be represented as an
\(\F_{q_0}\)-linear combination of at most two elements of \(U\).
\end{claim}

\begin{proof}[Proof of Claim~\ref{claim1}]
First note that \(\F_q\cap U=\{u\in\F_q:~u^2=1\} \subseteq\F_{q_0}.\)
Indeed, for \(u\in\F_q\) one has \(u^q=u\), while \(u\in U\) means \(u^{q+1}=1\), and hence \(u^2=1\). Thus, \(\F_q\cap U=\{1\}\) in characteristic two and \(\F_q\cap U=\{1,-1\}\) in odd characteristic.
Suppose first that \(a=c_1u\) for some \(c_1\in\F_{q_0}^*\) and \(u\in U\). Since \(a,c_1\in\F_q\), we have \(u=a/c_1\in\F_q\cap U\subseteq\F_{q_0}\), and hence \(a\in\F_{q_0}\), a contradiction. Thus, \(a\) has no one-column representation.

Now suppose that \(a=c_1u+c_2v\), where \(c_1,c_2\in\F_{q_0}^*\) and \(u,v\in U\). We may assume both coefficients are nonzero, since otherwise this reduces to the one-column case. Using \(u^q=u^{-1}\) and \(\N(a-c_1u)=\N(c_2v)=c_2^2,\) we obtain
\(a^2+c_1^2-ac_1(u+u^{-1})=c_2^2.\)
After multiplying by \(u\) and dividing by \(ac_1\ne0\), this gives
\begin{equation}\label{eq:~weighted-quadratic}
 u^2-\frac{a^2+c_1^2-c_2^2}{ac_1}u+1=0.
\end{equation}

To treat both characteristics uniformly, set \(z=ac_1u\). Then quation ~\eqref{eq:~weighted-quadratic} becomes
$
 z^2-(a^2+c_1^2-c_2^2)z+a^2c_1^2=0.
$
By assumption, all coefficients belong to \(\F_{q_0^s}\). Since \(\F_q\) is the quadratic extension of \(\F_{q_0^s}\), every quadratic polynomial over \(\F_{q_0^s}\) splits over \(\F_q\). Hence \(z\in\F_q\), including the repeated-root case. It follows that \(u=z/(ac_1)\in\F_q\cap U\), and then \(v=(a-c_1u)/c_2\in\F_q\cap U\) as well. Thus, \(u,v\in\F_{q_0}\), which forces \(a\in\F_{q_0}\), again a contradiction.
\end{proof}

We now construct subspaces all of whose nonzero elements satisfy Claim~\ref{claim1}.

\begin{claim}\label{claim2}
If \(q_0\) is even, then for every \(1\le t\le s-1\) there exists a \(t\)-dimensional subspace \(W\le\F_q\) whose nonzero elements satisf Claim~\ref{claim1}. If \(q_0\) is odd, the same conclusion holds for every \(1\le t\le s\).
\end{claim}

\begin{proof}[Proof of Claim~\ref{claim2}]
Suppose first that \(q_0\) is even. Choose an \(\F_{q_0}\)-linear complement \(W_0\) of \(\F_{q_0}\) in \(\F_{q_0^s}\), so that \(\F_{q_0^s}=\F_{q_0}\oplus W_0\) and \(\dim_{\F_{q_0}}W_0=s-1\). Every nonzero \(a\in W_0\) lies outside \(\F_{q_0}\) and satisfies \(a^2\in\F_{q_0^s}\). Hence Claim~\ref{claim1} applies to every nonzero element of \(W_0\). For \(1\le t\le s-1\), any \(t\)-dimensional subspace \(W\le W_0\) therefore has the required property.

Now suppose that \(q_0\) is odd. Choose a nonsquare \(d\in\F_{q_0^s}^*\), and let \(b\in\F_q\) satisfy \(b^2=d\). Since \(d\) is nonsquare in \(\F_{q_0^s}\), we have \(b\notin\F_{q_0^s}\). The space \(W_0=b\F_{q_0^s}\) is an \(s\)-dimensional \(\F_{q_0}\)-subspace of \(\F_q\), and \(W_0\cap\F_{q_0}=\{0\}\). Indeed, if \(by\in\F_{q_0}^*\) for some \(y\in\F_{q_0^s}^*\), then \(b=(by)/y\in\F_{q_0^s}\), a contradiction. Moreover, for \(a=by\in W_0\) we have \(a^2=dy^2\in\F_{q_0^s}\). Thus, every nonzero element of \(W_0\) satisfies Claim~\ref{claim1}. Hence, for \(1\le t\le s\), any \(t\)-dimensional subspace \(W\le W_0\) has the desired property.
\end{proof}

It remains to convert this pointwise obstruction into a lower bound on the size of a common cover. Let \(W\) be a \(t\)-dimensional subspace supplied by Claim~\ref{claim2}, and let \(I=\{u_1,\ldots,u_\ell\}\subseteq U\) be any common cover of \(W\). Thus, \(W\subseteq\spnq{I}\). Define the map  $\sigma$ from $\F_{q_0}^{\ell}$ to $\F_{q^2}$ such that 
\(\sigma(e_1,\ldots,e_\ell) =\sum_{j=1}^{\ell}e_ju_j.\)
Choose a basis \(w_1,\ldots,w_t\) of \(W\), together with preimages \(e^{(1)},\ldots,e^{(t)}\in\F_{q_0}^{\ell}\) satisfying \(\sigma(e^{(i)})=w_i\). Then
\(D=\spnq{e^{(1)},\ldots,e^{(t)}} \le\F_{q_0}^{\ell}\)
has dimension \(t\), and the restriction \(\sigma|_D:~D\to W\) is an isomorphism. Indeed, any linear relation among the \(e^{(i)}\)'s would map to the same relation among the basis vectors \(w_i\).

For every nonzero \(e\in D\), the vector \(\sigma(e)\) is therefore a nonzero element of \(W\). If \(e\) had Hamming weight at most two, then \(\sigma(e)\) would be an \(\F_{q_0}\)-linear combination of at most two columns from \(I\), contradicting Claim~\ref{claim1}. Hence every nonzero codeword of \(D\) has Hamming weight at least three, so \(D\) is a linear \([\ell,t,\ge3]_{q_0}\) code.
By the Hamming bound~\cite[Theorem~1.12.1]{HuffmanPless2003}, we have \(q_0^t\bigl(1+(q_0-1)\ell\bigr)\le q_0^\ell,\) or equivalently \(q_0^{\ell-t}\ge1+(q_0-1)\ell\). By the definition of \(h_{q_0}(t)\), we obtain \(\ell\ge h_{q_0}(t)\).
Since \(I\) was an arbitrary common cover of \(W\), we have \(\tau_U(W)\ge h_{q_0}(t),\) and hence, by~\eqref{eq-gcv3}, we arrive at 
\(\rho_t(\mathcal C_{2s}(q_0)) \ge\tau_U(W) \ge h_{q_0}(t).\)
This proves~\eqref{eq:~h-lower} in the stated ranges.
\end{proof}

\begin{remark}\label{rem:~short-code-lifting}
The lifting argument in the last part of the proof does not depend on the assumption that \(m\) is even. More generally, let \(W\le\F_{q^2}\) be any \(t\)-dimensional \(\F_{q_0}\)-subspace, and suppose that every nonzero element of \(W\) requires at least \(d\) columns of \(U\).
If \(W\) is covered by \(\ell\) columns, then the same construction produces a linear
\([\ell,t,\ge d]_{q_0}\)
code. Thus, any lower bound on the length of such a linear code gives a corresponding lower bound on \(\tau_U(W)\). We will use this observation again in the binary case.
\end{remark}

\subsection{Upper bound}

For the upper bound, we do not need to prescribe a basis of the target
subspace in advance. Instead, we first find a column that occurs as a
common summand in the representations of sufficiently many vectors,
and then choose a basis among those vectors.

For
\(a\in\F_{q^2}\setminus(\F_{q_0}^{*}U\cup\{0\})\), define
\begin{equation}\label{eq:~anchor}
 A(a)=\{x\in U:~a=x+y+z\text{ for some }y,z\in U\}.
\end{equation}
Thus, \(A(a)\) consists of the columns that may serve as a prescribed
common summand in a three-column representation of \(a\).
Vectors in the scalar cone \(\F_{q_0}^{*}U\) are excluded because
they already admit a one-column representation.

\begin{proposition}\label{thm:~basis-selection}
Let \(1\le t\le2m\), and suppose that
\(|A(a)|\ge L\) for every
\(a\in\F_{q^2}\setminus(\F_{q_0}^{*}U\cup\{0\})\).
If
\begin{equation}\label{eq:~abstract-incidence}
 (q_0^t-1)L>(q+1)(q_0^{t-1}-1),
\end{equation}
then
\[ \rho_t\bigl(\mathcal C_m(q_0)\bigr)\le2t+1.\]
\end{proposition}

\begin{proof}
We first consider a \(t\)-dimensional subspace \(W\le\F_{q^2}\) with \(W\cap\F_{q_0}^{*}U=\varnothing\). For each \(x\in U\), let \(B_x=\{a\in W\setminus\{0\}:~x\in A(a)\},\) and let \(G=\{x\in U:~\spnq{B_x}=W\}\). If \(x\notin G\), then \(B_x\) is contained in a proper subspace of \(W\), so \(|B_x|\le q_0^{t-1}-1\). For \(x\in G\), we use only the trivial bound \(|B_x|\le q_0^t-1\).
We now count the incidences \(\mathcal I=\{(a,x):~a\in W\setminus\{0\},\,x\in A(a)\}\) in two ways. Since every nonzero element of \(W\) lies outside \(\F_{q_0}^{*}U\),
\[ |\mathcal I| =\sum_{a\in W\setminus\{0\}}|A(a)| \ge(q_0^t-1)L.\]
On the other hand,
\[ |\mathcal I| =\sum_{x\in U}|B_x| \le(q+1-|G|)(q_0^{t-1}-1)+|G|(q_0^t-1).\]
Hence \((q_0^t-1)L \le (q+1)(q_0^{t-1}-1) +|G|q_0^{t-1}(q_0-1).\) By~\eqref{eq:~abstract-incidence}, this forces \(|G|\ge1\). Choose \(x\in G\). Since \(B_x\) spans \(W\), it contains an
\(\F_{q_0}\)-basis \(a_1,\ldots,a_t\) of \(W\). For each \(i\), there exist \(y_i,z_i\in U\) such that \(a_i=x+y_i+z_i\). Therefore, 
\(W\subseteq \spnq{x,y_1,z_1,\ldots,y_t,z_t},\) and hence \(\tau_U(W)\le2t+1\).

It remains to remove the assumption \(W\cap\F_{q_0}^{*}U=\varnothing\). First observe that~\eqref{eq:~abstract-incidence} automatically holds in every smaller positive dimension. Indeed, the ratio \(\frac{q_0^{j-1}-1}{q_0^j-1}\) is increasing in \(j\). 
Thus, \(\frac{L}{q+1}> \frac{q_0^{t-1}-1}{q_0^t-1}\) implies the corresponding inequality with \(j\) in place of \(t\) for every \(1\le j\le t\).
We now prove by induction on \(j\) that every \(j\)-dimensional subspace of \(\F_{q^2}\), for \(1\le j\le t\), can be covered by at most \(2j+1\) columns of \(U\). For \(j=1\), a line meeting \(\F_{q_0}^{*}U\) is covered by one column, while a line disjoint from \(\F_{q_0}^{*}U\) is covered by at most three columns by the preceding incidence argument.
Let \(2\le j\le t\), and assume the assertion holds in dimension \(j-1\). Let \(W\le\F_{q^2}\) have dimension \(j\). If \(W\cap\F_{q_0}^{*}U=\varnothing\), the preceding argument gives \(\tau_U(W)\le2j+1\).
Otherwise, choose \(0\ne a\in W\cap\F_{q_0}^{*}U\), and write \(a=\lambda u\) with \(\lambda\in\F_{q_0}^{*}\) and \(u\in U\). Since \(W\) is \(\F_{q_0}\)-linear, \(u\in W\).
Choose a complement \(W'\) such that \(W=\F_{q_0}u\oplus W'\). By the induction hypothesis, \(\tau_U(W')\le2j-1\), and adjoining the column \(u\) gives \(\tau_U(W)\le2j\).
Taking \(j=t\), every \(t\)-dimensional subspace of \(\F_{q^2}\) can be covered by at most \(2t+1\) columns of \(U\). By the subspace formulation~\eqref{eq-gcv3}, we have
\(\rho_t\bigl(\mathcal C_m(q_0)\bigr)\le2t+1.\)
\end{proof}

We first recall the two-summand criterion for the norm-one subgroup. Let \(b\in\F_{q^2}^*\) and put \(c=\N(b)\). It follows from \cite[Lemmas~6--8]{XiongYan2026} that
\begin{equation}\label{eq:~pair-tests}
 b\in U+U\quad\Longleftrightarrow\quad
 \begin{cases}
 \Tr_{\F_q/\Ftwo}(c^{-1})=1,&q\text{ even},\\
 \chi(1-4/c)=-1\text{ or }c=4,&q\text{ odd}.
 \end{cases}
\end{equation}
Here $U+U=\{u+v:~u,v\in U\}$ allows $u=v$. For nonzero $b$, the unordered pair of summands, if it exists, is unique, with a repeated pair allowed in odd characteristic.
Indeed, $u+v=b$ implies $uv=b/b^q$, so $u/b$ and $v/b$ are the roots of $Y^2-Y+\N(b)^{-1}$. If this polynomial is irreducible, its roots satisfy $y^q=1-y$, and $\N(by)=\N(b)y(1-y)=1$. Distinct roots in $\F_q$ cannot yield units:~ $\N(b)y^2=\N(b)y(1-y)=1$ would force $2y=1$.
The remaining case is a repeated root, which occurs precisely when $q$ is odd and $\N(b)=4$; its summands are $b/2,b/2$. This also checks the normalization and uniqueness needed later.
For even $q$, with the absolute trace defined at the beginning of Section~\ref{sec:~bounds}, we will also use the form
\begin{equation}\label{eq:~pairtrace}
 b\in U+U\ \Longleftrightarrow\ \tr(\N(b)^{-1})=1.
\end{equation}

To convert sums over \(U\) into rational-point counts, we use a standard parametrization of the norm-one subgroup. Choose \(\eta\in\F_{q^2}\setminus\F_q\) and define
\begin{equation}\label{eq:~param}
 X(T)=\frac{T-\eta}{T-\eta^q}~{\rm for}~T\in\F_q,
 \qquad X(\infty)=1.
\end{equation}
Here \(\PP^1(\F_q)\) denotes the projective line over \(\F_q\). Its \(\F_q\)-rational points may be identified with \(\F_q\cup\{\infty\}\), where \(T\in\F_q\) represents the projective point \([T:~1]\) and \(\infty\) represents \([1:~0]\).

For \(T\in\F_q\), since \(T^q=T\), we have \(X(T)^q =\frac{T-\eta^q}{T-\eta} =X(T)^{-1},\) and hence \(X(T)^{q+1}=1\). Thus, \(X(T)\in U\). Conversely, if \(x\in U\setminus\{1\}\), solving \(x=X(T)\) gives \(T=\frac{x\eta^q-\eta}{x-1}.\)
Using \(x^q=x^{-1}\), one checks that \(T^q=T\), and hence  \(T\in\F_q\). The remaining element \(x=1\) corresponds to \(T=\infty\). Consequently, \(X\) induces a bijection between \(\PP^1(\F_q)\) and \(U\).
More generally, \(X(T)\) is an invertible M\"obius transformation of \(\PP^1\) over \(\F_{q^2}\). Hence it preserves the orders of zeros and poles of rational functions. In particular, a point \(x_0\in\PP^1(\F_{q^2})\) has an \(\F_q\)-rational preimage under \(X\) if and only if \(x_0\in U\). We will use this observation below to determine whether poles of the transformed rational functions occur at \(\F_q\)-rational points.
Finally, the \(q\)-Frobenius acting on the coefficients fixes the indeterminate \(T\) and sends
$ X(T)=\frac{T-\eta}{T-\eta^q}$ to $\frac{T-\eta^q}{T-\eta}=X(T)^{-1}$.
This will allow us to recognize the rational functions obtained after the substitution \(x=X(T)\) as functions defined over \(\F_q\).

For $a\in\F_{q^2}\setminus(\F_{q_0}^*U\cup\{0\})$ and $x\in U$, the definition~\eqref{eq:~anchor} gives $x\in A(a)$ if and only if $a-x\in U+U$. We estimate the number of such $x$ using the preceding criterion.

\begin{proposition}\label{prop:~anchors}
For every \(a\in\F_{q^2}\setminus(\F_{q_0}^*U\cup\{0\})\), we have
\begin{equation}\label{eq:~anchor-lower}
 |A(a)|\ge\frac{q+1-2\sqrt q}{2}.
\end{equation}
\end{proposition}

\begin{proof}
Put \(c=\N(a)\). Since \(\F_{q_0}^*U = \{b\in\F_{q^2}^*:~\N(b)\in(\F_{q_0}^*)^2\},\) the assumption on \(a\) implies \(c\ne1\). In odd characteristic, we shall also use \(c\ne9\). Indeed, if the characteristic is not three, then \(9=3^2\) is a nonzero square in \(\F_{q_0}\), whereas in characteristic three this follows from \(c\ne0\).

Consider \[ f(Z)=c+1-a^qZ-\frac{a}{Z}\in\F_{q^2}(Z).\] 
For \(x\in U\), since \(x^q=x^{-1}\), \(f(x)=\N(a-x)\).
Moreover, \(a\notin U\), so \(a-x\ne0\) and hence \(f(x)\ne0\) for every \(x\in U\). Therefore,  we may apply the two-summand criterion to \(a-x\).
After the substitution \(Z=X(T)\), the function \(f(X(T))\) belongs to \(\F_q(T)\). Indeed, the \(q\)-Frobenius interchanges \(a\) and \(a^q\) and sends \(X(T)\) to \(X(T)^{-1}\), leaving the composite unchanged. The same is true of the rational functions derived from \(f(X(T))\) below.

By~\eqref{eq:~pairtrace}, we have \(x\in A(a)\) if and only if \(\tr\left(\frac1{f(x)}\right)=1\). Moreover, we have \(\frac{1}{f(x)} = \frac{x}{(x+a)(a^qx+1)}.\)
As a rational function of \(x\), this has two simple poles, at \(x=a\) and \(x=a^{-q}\). They are distinct because \(c\ne1\), and neither lies in \(U\). Since \(X(T)\) is a M\"obius transformation, \(1/f(X(T))\) also has exactly two simple poles, and neither lies above an \(\F_q\)-rational point under the parametrization~\eqref{eq:~param}.
Consider the Artin-Schreier curve
\[ C:~Y^2+Y=\frac1{f(X(T))}.\]
By~\eqref{eq:~AS-genus}, the two simple poles give \(g(C)=1\). They also ensure that \(C\) is geometrically irreducible. Moreover, the right-hand side has no pole at any point of \(\PP^1(\F_q)\).
For \(r\in\F_q\), the equation \(Y^2+Y=r\) has two \(\F_q\)-solutions when \(\tr(r)=0\), and none when \(\tr(r)=1\). Therefore, 
\[ \#C(\F_q) = q+1+ \sum_{x\in U} (-1)^{\tr(1/f(x))}.\]
Since \(g(C)=1\), it follows from the Hasse-Weil bound~\eqref{eq:~HW-genus-one} that
\[ \left| \sum_{x\in U} (-1)^{\tr(1/f(x))} \right| \le2\sqrt q.\]

On the other hand, \(A(a)\) is precisely the set of \(x\in U\) for which \(\tr(1/f(x))=1\). Hence \(q+1-|A(a)|\) elements of \(U\) have trace zero, while \(|A(a)|\) elements have trace one. Therefore, 
\[ \sum_{x\in U}(-1)^{\tr(1/f(x))}=(q+1-|A(a)|)-|A(a)|=q+1-2|A(a)|.\]
It follows that
\[2|A(a)| =q+1-\sum_{x\in U}(-1)^{\tr(1/f(x))} \ge q+1-2\sqrt q,\] 
which proves~\eqref{eq:~anchor-lower} in even characteristic.

Now suppose that \(q\) is odd. By~\eqref{eq:~pair-tests}, \(x\in A(a)\) precisely when \(\chi\!\left(1-\frac4{f(x)}\right)=-1,\) together with the repeated-root case \(f(x)=4\). Since \(f(x)\ne0\) on \(U\), then 
$\chi\!\left(1-\frac4{f(x)}\right) = \chi\bigl(f(x)(f(x)-4)\bigr).$ 
Set \(h(x)=f(x)(f(x)-4).\) The zeros of \(f\) are \(a\) and \(a^{-q}\), which are distinct because \(c\ne1\). The zeros of \(f-4\) are the roots of \(a^qx^2-(c-3)x+a=0,\) whose discriminant is \((c-1)(c-9)\). Since \(c\ne1,9\), these two roots are also distinct and nonzero. The zero sets of \(f\) and \(f-4\) are disjoint, so \(h\) has four simple zeros.
Both \(f\) and \(f-4\) have simple poles at \(x=0\) and \(x=\infty\). Hence \({\rm ord}_0(h)={\rm ord}_\infty(h)=-2,\) and hence \(h\) has poles of order two at these two points. Since the M\"obius transformation \(X(T)\) preserves orders, the same zero and pole orders occur for \(h(X(T))\). Consider the double cover
\[ C:~ Y^2=h(X(T)).\]
The four simple zeros are the four branch points of the quadratic cover, while the two poles of order two are unramified. It follows from~\eqref{eq:~RH-double} that \(2g(C)-2=-4+4=0,\) and hence \(g(C)=1\). Since \(h(X(T))\) has a simple zero, it is not a square, so \(C\) is geometrically irreducible.

The poles correspond to \(x=0\) and \(x=\infty\), neither of which lies in \(U\). Hence \(h(X(T))\) has no pole at any point of \(\PP^1(\F_q)\). For \(r\in\F_q\), the equation \(Y^2=r\) has \(1+\chi(r)\) solutions in \(\F_q\), where \(\chi(0)=0\). Therefore, 
\[\#C(\F_q) = q+1+ \sum_{x\in U}\chi(h(x)).\]
Since \(g(C)=1\), according to the Hasse-Weyl inequality~\eqref{eq:~HW-genus-one}, we can once again conclude that
\[ \left| \sum_{x\in U}\chi(h(x)) \right| \le2\sqrt q.\]

Let \(z=\#\{x\in U:~f(x)=4\},\) and let \(n_+\) and \(n_-\) denote the numbers of \(x\in U\) for which \(\chi(h(x))=1\) and \(-1\), respectively. Since \(f(x)\ne0\) on \(U\), the zeros of \(h\) on \(U\) are exactly the points counted by \(z\). Hence
\[ q+1=n_++n_-+z~{\rm and}~ \sum_{x\in U}\chi(h(x))=n_+-n_-.\]
By the two-summand criterion, \(A(a)\) consists of the \(n_-\) points with nonsquare value together with the \(z\) repeated-root points. Hence, \(|A(a)|=n_-+z\). By eliminating \(n_+\) and \(n_-\), we have
\[ 2|A(a)|=q+1-\sum_{x\in U}\chi(h(x))+z.\]
Since \(z\ge0\), the Hasse-Weil estimate yields \(2|A(a)|\ge q+1-2\sqrt q,\) and~\eqref{eq:~anchor-lower} follows. 
\end{proof}

Propositions~\ref{thm:~basis-selection} and~\ref{prop:~anchors} now reduce the proof of the general upper bound to checking a single explicit inequality.

\begin{proof}[Proof of Theorem~\ref{thm:~main-bounds}\textup{(i)}]
By Proposition~\ref{prop:~anchors}, we may take \(L=\frac{q+1-2\sqrt q}{2}\) in Proposition~\ref{thm:~basis-selection}. Hence it is enough to verify
$(q_0^t-1)\frac{q+1-2\sqrt q}{2} > (q+1)(q_0^{t-1}-1).$ 
After rearranging, this is equivalent to
\begin{equation}\label{eq:~upper-condition}
 (q+1)\bigl(q_0^t-2q_0^{t-1}+1\bigr)> 2(q_0^t-1)\sqrt q.
\end{equation}

\begin{itemize}
    \item {\bf Case 1:~ \(q_0=2\).} In this case~\eqref{eq:~upper-condition} becomes
\(q+1>2(2^t-1)\sqrt q.\)
If \(m\ge2t+2\), then \(\sqrt q=2^{m/2}\ge2^{t+1},\) and hence 
\(2(2^t-1)\sqrt q \le q-2\sqrt q <q+1.\)
Thus, \(\rho_t(\mathcal C_m(2))\le2t+1.\)

\item {\bf Case 2:~ \(q_0\ge3\) and \(2t+1\ge2m\).} The trivial dimension bound \(\rho_t(\mathcal C_m(q_0))\le 2m\) already gives the desired conclusion.

\item {\bf Case 3:~ \(q_0\ge3\) and \(2t+1<2m\).} Hence \(t\le m-1\) and \(m\ge2\). Put \(z=\sqrt q=q_0^{m/2}.\) The difference between the left- and right-hand sides of
\eqref{eq:~upper-condition} is
\begin{equation}\label{eq:~nonbinary-gap}
 q_0^{t-1} \bigl((q_0-2)(q+1)-2q_0z\bigr) +(q+1)+2z.
\end{equation}

\begin{itemize}
    \item {\bf Case 3.1:~ \(q_0\ge4\)}. Then \(z\ge4\) and \(q+1=z^2+1>4z \ge\frac{2q_0}{q_0-2}z.\) Thus, \((q_0-2)(q+1)-2q_0z>0,\) and~\eqref{eq:~nonbinary-gap} is positive.

    \item {\bf Case 3.2:~ \(q_0=3\) and \(m\ge4\).} Then \(z\ge9\), and \(q+1=z^2+1>6z.\)
Again, \((q_0-2)(q+1)-2q_0z=(q+1)-6z>0,\) so~\eqref{eq:~upper-condition} holds.

\item {\bf Case 3.3:~ \((q_0,m)=(3,2) and (3,3)\).} For \(m=2\), the only nontrivial order after using the dimension bound is \(t=1\); the incidence condition reduces to
\(L=\frac{10-6}{2}=2>0.\) For \(m=3\), it is enough to check \(t=2\), for which \eqref{eq:~upper-condition} becomes
\(28>4\sqrt{27}.\)
This is true. Since \(\frac{q_0^{j-1}-1}{q_0^j-1}\)
is increasing in \(j\), the order \(t=1\) condition follows as well.
\end{itemize}
\end{itemize}
In conclusion, we have completed the proof.
\end{proof}

The general upper bound can be improved for a fixed target subspace
by taking advantage of the columns that already lie inside it.

\begin{corollary}
Let \(q_0\ge3\), let \(W\le\F_{q^2}\) have dimension \(t\), and put
$
 r(W)=\dim_{\F_{q_0}}\spnq{U\cap W}.
$
If \(r(W)=t\), then \(\tau_U(W)=t\). Otherwise,
\begin{equation}\label{eq:~zetterberg-pointwise}
 \tau_U(W)
 \le
 \min\{2m,\,2t-r(W)+1\}.
\end{equation}
\end{corollary}

\begin{proof}
Let \(r=r(W)<t\), and choose \(r\) columns from \(U\cap W\) forming a basis of \(\spnq{U\cap W}\). Choose an \(\F_{q_0}\)-linear complement \(W'\) such that \(W=\spnq{U\cap W}\oplus W'\) and \(\dim W'=t-r\).
By Theorem~\ref{thm:~main-bounds}\textup{(i)}, the subspace \(W'\) can be covered by at most \(2(t-r)+1\) columns of \(U\). Together with the \(r\) internal columns already chosen, this gives
\[ \tau_U(W)\le r+2(t-r)+1=2t-r+1.\]
Combining this with the trivial bound \(\tau_U(W)\le2m\) proves \eqref{eq:~zetterberg-pointwise}. If \(r(W)=t\), then \(\spnq{U\cap W}=W\), so the internal-spanning criterion gives \(\tau_U(W)=t\).
\end{proof}

The lower bound in Theorem~\ref{thm:~main-bounds}\textup{(ii)}
also becomes exact whenever the ordinary covering radius is two.

\begin{corollary}\label{cor:~exact-2t}
Suppose that \(\rho_1(\Cm)=2\), \(2\le t\le m\), and
\(q\ge\Theta_{q_0,t}\). Then
\[
 \rho_t(\Cm)=2t.
\]
\end{corollary}

\begin{proof}
By the subadditivity in~\eqref{eq:~basic},
\(\rho_t(\Cm)\le t\rho_1(\Cm)=2t.\)
On the other hand, Theorem~\ref{thm:~main-bounds}\textup{(ii)} gives \(\rho_t(\Cm)\ge2t\). Hence equality holds.
\end{proof}

As a concrete specialization, the classical condition \(\rho_1(\Cm)=2\) holds when \(q_0\) is even and \(3\le m\le q_0/2\) is odd~\cite[Corollary~4]{XiongYan2026}.
Taking \(m=7\) and \(t=2\), we obtain \(\rho_2(\mathcal C_7(q_0))=4\) for every even \(q_0\ge16\). The ordinary-radius results in odd characteristic from~\cite[Corollary~2]{XiongYan2026} yield analogous specializations, with \(\delta=2\).

\section{The extremal case \(\rho_t(\Cm)=t\)}\label{sec:~high}

The smallest possible value of the $t^{\rm th}$ generalized covering radius is \(t\). By the characterization in Subsection~\ref{subsec:~covering-definitions}, equality \(\rho_t(\Cm)=t\) holds precisely when, for every \(t\)-dimensional subspace \(W\le\F_{q^2}\), the columns of \(U\) contained in \(W\) span \(W\). Equivalently, \(U\cap W\) cannot be contained in any hyperplane of \(W\).
This suggests estimating the relative intersection \(U\cap(W\setminus L)\), where \(L\) is a proper subspace of \(W\), rather than estimating \(U\cap W\) alone. We first establish a uniform estimate for such differences. The hyperplane case will then give the internal-spanning property, and hence the exact value \(\rho_t=t\).

We begin with the additive Fourier coefficients of the norm-one subgroup. Let \(\psi\) be the canonical nontrivial additive character of \(\F_q\). For \(b,x\in\F_{q^2}\), put
\(\Psi_b(x)=\psi\!\left(\Tr_{\F_{q^2}/\F_q}(bx)\right).\)
Define
\begin{equation*}
 \widehat U(b)=\sum_{u\in U}\Psi_b(u).
\end{equation*}
For \(b\ne0\), these Fourier coefficients can be expressed as Kloosterman sums with the following form 
\begin{equation}\label{eq:~fourier-bound}
 \widehat U(b)
 =-\sum_{z\in\F_q^*}
   \psi\!\left(z+\frac{\N(b)}{z}\right),
 \quad
 |\widehat U(b)|\le2\sqrt q.
\end{equation}

For completeness, we briefly verify the identity. Put \(c=\N(b)\). Multiplication by \(b\) maps \(U\) bijectively onto the norm-\(c\) fibre in \(\F_{q^2}^*\). For \(a\in\F_q\), let \(M(a)\) denote the number of \(x\in\F_{q^2}\) satisfying \(\Tr_{\F_{q^2}/\F_q}(x)=a\) and \(\N(x)=c\), and let \(R(a)\) denote the number of \(z\in\F_q^*\) satisfying \(z+c/z=a\).
Both quantities are governed by the quadratic \(Z^2-aZ+c\).
If this quadratic is irreducible over \(\F_q\), its two roots are interchanged by the \(q\)-Frobenius and both have trace \(a\) and norm \(c\), and hence \(M(a)=2\) and \(R(a)=0\). If it splits into two distinct roots in \(\F_q\), Frobenius fixes each root separately, so neither root has the other as its conjugate. Thus, \(M(a)=0\), whereas \(R(a)=2\). In the repeated-root case, the unique root is counted once by each quantity, so \(M(a)=R(a)=1\). Therefore,  \(M(a)+R(a)=2\) in all cases.

Since multiplication by \(b\) is a bijection from \(U\) onto the norm-\(c\) fibre, where \(c=\N(b)\), we have
\[\widehat U(b) =\sum_{u\in U}\psi\!\left(\Tr_{\F_{q^2}/\F_q}(bu)\right) =\sum_{a\in\F_q}M(a)\psi(a).\]
On the other hand, combining with the definition of \(R(a)\), we obtain 
\[ \sum_{a\in\F_q}R(a)\psi(a) =\sum_{z\in\F_q^*}\psi\!\left(z+\frac{c}{z}\right).\]
Thus, summing \(M(a)+R(a)=2\) against \(\psi(a)\) and using
\(\sum_{a\in\F_q}\psi(a)=0\), we obtain
\[ \widehat U(b) =-\sum_{z\in\F_q^*}\psi\!\left(z+\frac{N(b)}{z}\right).\]
The estimate \(|\widehat U(b)|\le2\sqrt q\) is then the classical Weil bound for Kloosterman sums~\cite{Weil1948}. We shall also use \(\widehat U(0)=|U|=q+1\).

We now apply this Fourier estimate to the difference of two nested subspaces.

\begin{proposition}
Let \(W\le\F_{q^2}\) have dimension \(t\), and let \(L<W\) have codimension \(j\) in \(W\). Put \(d=2m-t\). Then, for \(1\le j\le t\), we have
\begin{equation}\label{eq:~norm-flag}
 \left| \frac{q_0^{d+j}}{q_0^j-1} |U\cap(W\setminus L)|-(q+1) \right|
 \le 2(2q_0^d-1)\sqrt q.
\end{equation}
In particular, if \(L\) is a hyperplane of \(W\), then
\begin{equation}\label{eq:~relative-flat}
 |U\cap(W\setminus L)| \ge \frac{q_0-1}{q_0^{2m-t+1}}\Delta_t,
\end{equation}
where \(\Delta_t=q+1-2(2q_0^{2m-t}-1)\sqrt q.\)
\end{proposition}

\begin{proof}
We identify \(\F_{q^2}\) with a \(2m\)-dimensional vector space over \(\F_{q_0}\) and use the nondegenerate bilinear trace pairing \(\langle b,x\rangle = \Tr_{\F_{q^2}/\F_{q_0}}(bx).\) Let \(\psi_0\) be the canonical additive character of \(\F_{q_0}\). By trace transitivity, \(\psi_0(\langle b,x\rangle)=\Psi_b(x)\).
For an \(\F_{q_0}\)-subspace \(E\le\F_{q^2}\), let \(E^\perp\) denote its orthogonal complement with respect to this pairing. Character orthogonality gives
\begin{equation}\label{eq:~subspace-indicator}
 |E^\perp|\,\boldsymbol 1_E(x) = \sum_{b\in E^\perp}\psi_0(\langle b,x\rangle).
\end{equation}
Indeed, if \(x\in E\), then every term on the right is \(1\). If \(x\notin E\), the map \(b\mapsto\langle b,x\rangle\) is a nonzero \(\F_{q_0}\)-linear functional on \(E^\perp\), since \((E^\perp)^\perp=E\). Its values are therefore uniformly distributed over \(\F_{q_0}\), and the corresponding character sum vanishes.

Set \(A=W^\perp\) and \(B=L^\perp\). Since \(L\subset W\), we have \(A\subset B\). Moreover, \(|A|=q_0^d\) and \(|B|=q_0^{d+j}\), because \(\dim W=t\) and \(\dim L=t-j\). 
Applying~\eqref{eq:~subspace-indicator} with \(E=W\) and summing over \(x\in U\), we obtain
\(q_0^d |U\cap W| =\sum_{b\in A}\widehat U(b),\)
because \(A=W^\perp\) has cardinality \(q_0^d\). Similarly, 
\(q_0^{d+j}|U\cap L| =\sum_{b\in B}\widehat U(b).\)
Here we have used trace transitivity to identify \(\psi_0(\langle b,x\rangle)\) with \(\Psi_b(x)\).

We subtract these two identities after multiplying the first by \(q_0^j\). Since
\(|U\cap(W\setminus L)|=|U\cap W|-|U\cap L|\) and \(A\subset B\), this yields
\begin{align*}
 q_0^{d+j}|U\cap(W\setminus L)|
 &=q_0^j\sum_{b\in A}\widehat U(b) -\sum_{b\in B}\widehat U(b)\\
 &=q_0^j\sum_{b\in A}\widehat U(b) -\sum_{b\in A}\widehat U(b)
 - \left(\sum_{b\in B}\widehat U(b)-\sum_{b\in A}\widehat U(b)\right)\\
 &=(q_0^j-1)\sum_{b\in A}\widehat U(b) -\sum_{b\in B\setminus A}\widehat U(b)\\
 &=(q_0^j-1)(q+1)+(q_0^j-1)\sum_{b\in A\setminus\{0\}}\widehat U(b)
   -\sum_{b\in B\setminus A}\widehat U(b),
\end{align*}
where in the last equality we used \(\widehat U(0)=|U|=q+1\).

This subtraction is the key point. The first nonconstant sum contains \(q_0^d-1\) terms, whereas \(|B\setminus A| =q_0^{d+j}-q_0^d =(q_0^j-1)q_0^d.\)
Using~\eqref{eq:~fourier-bound}, the total contribution of the nonzero Fourier coefficients is therefore at most \(2(q_0^j-1)(2q_0^d-1)\sqrt q.\)
After dividing by \(q_0^j-1\), we obtain \eqref{eq:~norm-flag}. In particular, the normalized error term is independent of the codimension \(j\).

Finally, if \(L\) is a hyperplane of \(W\), then \(j=1\). Since \(d=2m-t\), the lower bound in~\eqref{eq:~norm-flag} gives
\[|U\cap(W\setminus L)|\ge
 \frac{q_0-1}{q_0^{2m-t+1}} \left(q+1-2(2q_0^{2m-t}-1)\sqrt q\right),\]
which is exactly~\eqref{eq:~relative-flat}.
\end{proof}

The proposition has a direct geometric consequence. If \(\Delta_t>0\), then for every \(t\)-dimensional subspace \(W\) and every hyperplane \(L<W\), the set \(U\cap(W\setminus L)\) is nonempty. Thus, \(U\cap W\) is not contained in any hyperplane of \(W\), and hence  spans \(W\). By \eqref{eq:~internal-spanning}, this is exactly the condition for \(\rho_t(\Cm)=t\).

\begin{proof}[Proof of Theorem~\ref{thm:~main-high}]
Let \(W\le\F_{q^2}\) be any \(t\)-dimensional subspace. If \(\Delta_t>0\), then by~\eqref{eq:~relative-flat}, we have \(U\cap(W\setminus L)\ne\varnothing\) for every hyperplane \(L<W\). Suppose that \(\spnq{U\cap W}\) is a proper subspace of \(W\). Then it would be contained in some hyperplane \(L<W\), forcing \(U\cap(W\setminus L)=\varnothing\), a contradiction. Hence \(\spnq{U\cap W}=W\). Since \(\dim W=t\), the set \(U\cap W\) contains a basis of \(W\), so \(\tau_U(W)\le t\). Conversely, fewer than \(t\) vectors cannot span a \(t\)-dimensional space, and hence \(\tau_U(W)\ge t\). Therefore,  \(\tau_U(W)=t\). As \(W\) is arbitrary,~\eqref{eq-gcv3} yields \(\rho_t(\Cm)=t.\)

It remains to derive the stated explicit ranges. Let \(m=2s\), with \(s\ge2\), and set \(d=4s-t\). Since \(2m-t=d\) and \(\sqrt q=q_0^{m/2}\), we may rewrite
\begin{equation*}
 \Delta_t= q\bigl(1-4q_0^{d-m/2}\bigr)+2\sqrt q+1.
\end{equation*}

If \(t\ge3s+2\), then \(d\le s-2=m/2-2\), and hence \(q_0^{d-m/2}\le q_0^{-2}\). Therefore, 
\[\Delta_t\ge q\left(1-\frac4{q_0^2}\right)+2\sqrt q+1>0\]
for every \(q_0\ge2\). Thus
\(\rho_t(\mathcal C_{2s}(q_0))=t\) throughout
\(3s+2\le t\le4s\).

If \(q_0\ge4\) and \(t\ge3s+1\), then \(d\le s-1=m/2-1\), so \(q_0^{d-m/2}\le q_0^{-1}\). Hence
\[\Delta_t \ge q\left(1-\frac4{q_0}\right)+2\sqrt q+1>0.\]
Consequently, \(\rho_t(\mathcal C_{2s}(q_0))=t\) throughout \(3s+1\le t\le4s\) whenever \(q_0\ge4\).
\end{proof}

We next treat odd extension degree. The same margin condition in Theorem~\ref{thm:~main-high} gives explicit ranges in this case as well.

\begin{corollary}\label{cor:~odd-high}
Let \(m=2s+1\), with \(s\ge1\). Then
\begin{equation}\label{eq:~odd-high}
 \rho_t(\mathcal C_{2s+1}(q_0))=t
 \quad\text{whenever}\quad
 \begin{cases}
 3s+2\le t\le4s+2,& q_0\ge16,\\
 3s+3\le t\le4s+2,& 3\le q_0<16,\\
 3s+4\le t\le4s+2,& q_0=2,
 \end{cases}
\end{equation}
where only nonempty ranges are asserted.
\end{corollary}

\begin{proof}
For \(m=2s+1\), the margin can be written as
\(\Delta_t = q\left(1-4q_0^{\,2m-t-m/2}\right)+2\sqrt q+1.\)
Since \(2m-t-\frac m2 =3s+\frac{3}{2}-t\), the three lower bounds on \(t\) in~\eqref{eq:~odd-high} give \(2m-t-\frac{m}{2} \le -\frac{1}{2}\), \(-\frac{3}{2}\), and \(-\frac{5}{2}\), respectively. Hence
\[\Delta_t \ge
 \begin{cases}
 q\left(1-\dfrac4{q_0^{1/2}}\right)+2\sqrt q+1,
     & t\ge3s+2,\\[6pt]
 q\left(1-\dfrac4{q_0^{3/2}}\right)+2\sqrt q+1,
     & t\ge3s+3,\\[6pt]
 q\left(1-\dfrac4{q_0^{5/2}}\right)+2\sqrt q+1,
     & t\ge3s+4.
 \end{cases}\]
The leading factor is nonnegative when \(q_0\ge16\) in the first case, when \(q_0\ge3\) in the second, and when \(q_0=2\) in the third. Since \(2\sqrt q+1>0\), we have \(\Delta_t>0\) throughout the stated ranges. The conclusion now follows from Theorem~\ref{thm:~main-high}.
\end{proof}

\section{Binary Zetterberg codes}\label{sec:~binary}

In this section, we focus on the binary case and determine several additional exact generalized covering radii. The additive structure of the norm-one subgroup over \(\F_2\) allows us to sharpen the general bounds obtained in Section~\ref{sec:~bounds}.
For the even-degree family \(\mathcal Z_s=\mathcal B_{2s}\), Theorem~\ref{thm:~main-bounds}
\textup{(i)} and \textup{(iii)} give
\begin{equation}\label{eq:~main-low}
 h_2(t)\le\rho_t(\mathcal Z_s)\le 2t+1 \quad(1\le t\le s-1).
\end{equation}
In particular,
\begin{equation}\label{eq:~third-range}
 6\le\rho_3(\mathcal Z_s)\le7
 \quad(s\ge4).
\end{equation}
Whenever the hypothesis of Theorem~\ref{thm:~main-bounds}\textup{(ii)} is satisfied, the lower bound \(h_2(t)\) in~\eqref{eq:~main-low} can be strengthened to \(2t\).

\subsection{The second radius in every extension degree}

Let \(q=2^m\), with \(m\ge1\), and retain the notation \(\tr=\Tr_{\F_q/\F_2}\) for the absolute trace. We first sharpen the common-summand estimate from Proposition~\ref{prop:~anchors}. This will provide the missing upper bounds in the two small extension degrees not covered by Theorem~\ref{thm:~main-bounds}\textup{(i)}. Then we then count two-dimensional deep-syndrome subspaces to obtain the corresponding lower bound.

\begin{lemma}\label{lem:~anchor-divisibility}
For every \(a\in\F_{q^2}\setminus(U\cup\{0\})\), the set \(A(a)\) defined in~\eqref{eq:~anchor} has cardinality divisible by three. Consequently,
\begin{equation}\label{eq:~integer-anchor}
 |A(a)|\ge L_m:= 3\left\lceil \frac{q+1-\lfloor2\sqrt q\rfloor}{6} \right\rceil.
\end{equation}
In particular, \(L_4=6\) and \(L_5=12\).
\end{lemma}

\begin{proof}
Consider a representation \(a=x+y+z\) with \(x,y,z\in U\). The three summands are pairwise distinct. Indeed, if two of them coincided, then in characteristic two they would cancel, leaving \(a\in U\), contrary to the assumption.
Let \(\mathscr T(a)\) be the set of unordered three-element subsets of \(U\) whose sum is \(a\). Two distinct elements of \(\mathscr T(a)\) cannot share an element. For if two such triples had a common element \(x\), then the remaining two elements in each triple would give two distinct unordered representations of \(a+x\) as a sum of two elements of \(U\), contradicting the uniqueness in the two-summand criterion~\eqref{eq:~pair-tests}.
Thus, the triples in \(\mathscr T(a)\) are pairwise disjoint, and their union is precisely \(A(a)\). Hence \(|A(a)|=3|\mathscr T(a)|.\)
By Proposition~\ref{prop:~anchors}, \(6|\mathscr T(a)|=2|A(a)| \ge q+1-2\sqrt q.\)
Since the left-hand side is an integer multiple of six, rounding up to the next such multiple gives~\eqref{eq:~integer-anchor}. The values \(L_4=6\) and \(L_5=12\) follow immediately.
\end{proof}

We next construct and count two-dimensional subspaces that provide the lower bound for the second generalized covering radius. For \(a\in\F_q^*\), one has \(\N(a)=a^2\) and \(\tr(a^{-2})=\tr(a^{-1})\). Hence the two-summand criterion \eqref{eq:~pair-tests} reduces to
\(a\in U+U\) if and only if \(\tr(a^{-1})=1\).
Moreover, \(\F_q\cap U=\{1\}\). Therefore,  every element of
\[H_m:= \{a\in\F_q^*\setminus\{1\}:~       \tr(a^{-1})=0\}\]
admits neither a one-column nor a two-column representation. Proposition~\ref{prop:~anchors} provides a three-column representation, so every element of \(H_m\) has covering cost exactly three.

Consequently, if \(W\le\F_q\) is two-dimensional and \(W\setminus\{0\}\subseteq H_m\), then all three nonzero elements of \(W\) are deep syndromes. We denote the number of such subspaces
by
\begin{equation*}
 D_m =\#\left\{ W\le\F_q:~ \dim_{\F_2}W=2,\; W\setminus\{0\}\subseteq H_m \right\}.
\end{equation*}
We shall determine \(D_m\) explicitly. For this purpose, define the binary Kloosterman sum
\begin{equation*}
 K_m = \sum_{z\in\F_q^*} (-1)^{\tr(z+z^{-1})}.
\end{equation*}

Note that generalized covering radii of binary Zetterberg codes have not previously been determined. A key ingredient in our lower bound is an exact enumeration of the two-dimensional syndrome subspaces all of whose nonzero vectors are deep.

\begin{proposition}\label{prop:~deep-plane-count}
For \(m\ge3\),
\begin{equation}\label{eq:~deep-plane-count}
 D_m=
 \begin{cases}
 \dfrac{(q-2)(q-8)}{48},&m\text{ odd},\\[6pt]
 \dfrac{q^2-14q+58-6K_m}{48},&m\text{ even}.
 \end{cases}
\end{equation}
\end{proposition}

\begin{proof}
We first count the two-dimensional subspaces
\(W\le\F_q\) for which every nonzero element has inverse trace zero,
temporarily allowing \(1\in W\).
Let \((a,b)\) be an ordered basis of such a subspace. Since its three nonzero elements are \(a,b,a+b\), we require \(\tr(a^{-1})=\tr(b^{-1}) =\tr((a+b)^{-1})=0.\)
Let \(z=a^{-1}+b^{-1}\) and \(x=\frac{a^{-1}}{z}\). Since \(a,b,a+b\) are all nonzero, we have \(z\ne0\) and \(x\notin\{0,1\}\). Conversely, \(a=(zx)^{-1}\) and \(b=(z(x+1))^{-1}\), and, in characteristic two, \((a+b)^{-1}=zx(x+1).\)
Thus, \((a,b)\mapsto(z,x)\) is a bijective change of variables, and the three inverse-trace conditions become
\begin{equation}\label{eq:~plane-linear-conditions}
 \tr(z)=0,\quad
 \tr(zx)=0,\quad
 \tr\bigl(z(x^2+x)\bigr)=0.
\end{equation}

\begin{claim}\label{claim3}
For fixed \(z\) such that \(z\ne 0\), \(z\ne1\), and \(\tr(z)=0\), then there are \(\frac{q}{4}-2\) choices of \(x\) such that \(\tr(zx)=0\)  and \(\tr\bigl(z(x^2+x)\bigr)=0\).
\end{claim}

\begin{proof}[Proof of Claim \ref{claim3}]
Once \(z\) is fixed, the last two conditions in \eqref{eq:~plane-linear-conditions} are linear in \(x\). Indeed, Frobenius invariance of the absolute trace gives \(\tr(zx^2)=\tr(z^{2^{m-1}}x)\),
so these two linear functionals are represented by \(z\) and \(z^{2^{m-1}}+z\) with respect to the nondegenerate trace pairing. Since \(z\ne1\), \(z^{2^{m-1}}+z\ne0\); otherwise \(z^{2^{m-1}}=z\), and since \(\gcd(m,m-1)=1\), this would imply \(z\in\F_2^*\), hence \(z=1\). The two representing coefficients \(z\) and \(z^{2^{m-1}}+z\) are also distinct, since equality would force \(z^{2^{m-1}}=0\). Thus, the two linear functionals are independent, and their common kernel has \(q/4\) elements. Both \(0\) and \(1\) belong to this kernel, so there are \(q/4-2\) admissible choices of \(x\).
\end{proof}

There are \(q/2-1\) nonzero elements \(z\in\F_q\) with \(\tr(z)=0\). If \(m\) is odd, then \(\tr(1)=1\), so none of these values is equal to \(1\). Hence the number of ordered bases is \((q/2-1)(q/4-2)\). Since every two-dimensional binary space has \(|\operatorname{GL}_2(2)|=6\) ordered bases, we obtain
\[ D_m =\frac{(q/2-1)(q/4-2)}6 =\frac{(q-2)(q-8)}{48}.\]
No further correction is needed, because \(1\) itself has inverse trace one when \(m\) is odd.

Now suppose that \(m\) is even. Then \(\tr(1)=0\), so the exceptional value \(z=1\) occurs. In this case \(z^{2^{m-1}}+z=0\), and only the nonzero linear condition \(\tr(zx)=0\) remains.
Its kernel has \(q/2\) elements, of which \(0\) and \(1\) are inadmissible. Thus, \(z=1\) contributes \(q/2-2\) choices of \(x\), whereas each of the remaining \(q/2-2\) nonzero trace-zero values of \(z\) contributes \(q/4-2\) choices. Dividing by six, we find that the number of two-dimensional subspaces whose nonzero elements all have inverse trace zero is
\begin{equation*}
 \widetilde D_m = \frac{(q/2-2)(q/4-2)+(q/2-2)}6 = \frac{(q-4)^2}{48}.
\end{equation*}

The quantity \(\widetilde D_m\) still includes the subspaces containing \(1\), which must be removed in the definition of \(D_m\). Such a subspace has nonzero part \(\{1,b,1+b\}\), with \(b\notin\{0,1\}\), and the remaining conditions are
\(\tr(b^{-1})=\tr((1+b)^{-1})=0.\)
Put \(v=1+b^{-1}\). Since \(m\) is even, \(\tr(1)=0\), and the two conditions above are equivalent to \(\tr(v)=\tr(v^{-1})=0.\)
Indeed, \(b^{-1}=v+1\) and
\((1+b)^{-1}=1+v^{-1}\).
Let \(\psi(w)=(-1)^{\tr(w)}\). Character orthogonality gives
\begin{align*}
 \#\{v\ne0:~\tr(v)=\tr(v^{-1})=0\}
 &=\frac14\sum_{v\ne0}(1+\psi(v))(1+\psi(v^{-1}))\\
 &= \frac{q-3+K_m}{4},
\end{align*}
because
\(\sum_{v\ne0}\psi(v)=\sum_{v\ne0}\psi(v^{-1})=-1\).

The value \(v=1\) is included in this count but does not correspond
to an admissible \(b\). Hence there are
\((q-7+K_m)/4\) admissible values of \(v\).
Each two-dimensional subspace is counted twice, once by \(b\) and
once by \(1+b\). Therefore,  the number of subspaces to subtract is
\((q-7+K_m)/8\). It follows that
\[
 D_m
 =
 \widetilde D_m-\frac{q-7+K_m}{8}
 =
 \frac{q^2-14q+58-6K_m}{48},
\]
which proves~\eqref{eq:~deep-plane-count}.
\end{proof}

\begin{remark}
The preceding argument also gives a direct construction of a deep two-dimensional subspace. Choose \(z\in\F_q^*\) such that \(\tr(z)=0\) and \(\tr(z^{-1})=1\). By character orthogonality,
\[\#\{z\ne0:~\tr(z)=0,\ \tr(z^{-1})=1\} = \frac{q-1-K_m}{4},\]
which is positive for \(m\ge4\) by the Weil bound. For such a \(z\), the two linear conditions in \eqref{eq:~plane-linear-conditions} are independent, so their common kernel has \(q/4\ge4\) elements. Choose \(x\notin\{0,1\}\) in this kernel, and set \(a=(zx)^{-1}\) and \(b=(z(x+1))^{-1}.\) Then \(a,b,a+b\) all have inverse trace zero.
It remains to verify that \(\langle a,b\rangle_{\F_2}\) does not contain \(1\). If \(m\) is odd, this is automatic because \(\tr(1)=1\). Suppose therefore that \(m\) is even. If \(a=1\), then \(x=z^{-1}\), and the third condition in \eqref{eq:~plane-linear-conditions} would give \(\tr(z(x^2+x)) =\tr(z^{-1}+1)=1,\) a contradiction. The case \(b=1\) is analogous. Finally, if \(a+b=1\), then \(x^2+x=z^{-1}\), but the equation \(X^2+X=z^{-1}\) has no solution in \(\F_q\) because \(\tr(z^{-1})=1\) (see \cite{BRS1967}). Thus, \(\langle a,b\rangle_{\F_2}\) is a deep two-dimensional subspace, and hence is counted by \(D_m\).
\end{remark}

\begin{proof}[Proof of Theorem~\ref{thm:~binary-main}\textup{(i)}]
We first consider \(m\ge4\). By Proposition~\ref{prop:~deep-plane-count}, there exists a two-dimensional subspace \(W\le\F_q\) whose three nonzero elements all require at least three columns. By the short-code lifting described in Remark~\ref{rem:~short-code-lifting}, any common cover of \(W\) by \(\ell\) columns gives a binary \([\ell,2,\ge3]\) code. Therefore,  it follows from the Hamming bound~\cite[Theorem~1.12.1]{HuffmanPless2003} that \(\ell\ge h_2(2)=5,\) and hence \(\rho_2(\mathcal B_m)\ge 5\).

For \(m\ge6\), Theorem~\ref{thm:~main-bounds}\textup{(i)}, applied with \(t=2\), gives the upper bound \(\rho_2(\mathcal B_m)\le5\).
It remains to treat \(m=4,5\), which lie outside the range of the general binary upper bound. We use Proposition~\ref{thm:~basis-selection} together with the refined estimate in Lemma~\ref{lem:~anchor-divisibility}. For \(t=2\), the incidence condition~\eqref{eq:~abstract-incidence} becomes \(3L_m>q+1\). When \(m=4\), this is \(18>17\), and when \(m=5\), it is \(36>33\). Thus, \(\rho_2(\mathcal B_m)\le5\) also in these two cases. Consequently, \(\rho_2(\mathcal B_m)=5\) for every \(m\ge 4\). In particular, no finite-field enumeration is needed for the cases \(m=4,5\).

We now consider the remaining small extension degrees. For \(m=1,2\), the codes \(\mathcal B_m\) have dimension one and lengths three and five, respectively. Since the sum of all elements of \(U\) is zero, the all-one word lies in the kernel of the parity-check matrix. As the kernel is one-dimensional, it generates the whole code. Hence \(\mathcal B_1\) and \(\mathcal B_2\) are the binary repetition codes of lengths three and five. 
Using \cite[Definition 8]{EFS2021} with \(t=2\), we can easily obtain \(\rho_2(\mathcal B_1)=2\) and \(\rho_2(\mathcal B_2)=3.\)

Finally, let \(m=3\). Then \(|U|=9\), and \(U\) contains the subgroup \(\F_4^*\) of order three. Its three cosets partition \(U\) into three disjoint triples, each having sum zero. Therefore,  the indicator vectors of these triples are codewords of \(\mathcal B_3\), and they are linearly independent because their supports are disjoint. Since \(\dim\mathcal B_3=9-6=3\), these three codewords form a basis, and
\(\mathcal B_3 \cong \operatorname{Rep}_2(3) \oplus \operatorname{Rep}_2(3) \oplus \operatorname{Rep}_2(3).\)
Note that \(\rho_2(\operatorname{Rep}_2(3))=\rho_2(\mathcal B_1)=2\). By \cite[Proposition~25]{EFS2021}, we have \(\rho_2(\mathcal B_3)=3\cdot2=6.\)
This completes the proof.
\end{proof}

\begin{remark}
The preceding count also shows that the extremal two-dimensional subspaces are far from unique. For \(m\ge4\), there are at least \((q+1)D_m\) distinct two-dimensional subspaces of \(\F_{q^2}\) with covering cost five. Indeed, if \(u\in U\), then multiplication by \(u\) permutes the column set \(U\), and hence \(\tau_U(uW)=\tau_U(W)\) for every subspace \(W\le\F_{q^2}\). Thus, each two-space \(W\) counted by \(D_m\) produces \(q+1\) two-spaces \(uW\), all having covering cost five.

These spaces are all distinct. If \(uW=vW'\) for \(u,v\in U\) and \(W,W'\le\F_q\), choose any nonzero vector in the common space. Then \(u/v\in\F_q\). Since also \(u/v\in U\), we have
\(\frac{u}{v}\in\F_q\cap U=\{1\},\) so \(u=v\), and hence \(W=W'\).
Therefore,  the number of distinct two-spaces attaining the second radius is at least \((q+1)D_m\). The first values \(D_4,D_5,D_6,D_7,D_8 = 2,15,69,315,1288\) follow directly from Proposition~\ref{prop:~deep-plane-count}. Independent checks for these small fields are recorded in the supplementary material.
\end{remark}

\subsection{The third radius on an infinite subfamily}

For the remainder of this subsection, let \(m=2s\), \(s\ge2\), and \(\mathcal Z_s=\mathcal B_{2s}\). The upper bound \(\rho_3(\mathcal Z_s)\le7\) is already given by \eqref{eq:~third-range}. We show that this bound is attained whenever \(4\mid s\). Thus, it remains to construct a three-dimensional subspace \(W\le\F_{q^2}\) satisfying \(\tau_U(W)>6\).
The proof proceeds in three steps. We first show that any six-column cover of a three-dimensional deep-syndrome subspace necessarily gives rise to the unique binary \([6,3,3]\) code, reducing the covering problem to seven affine planes in \(W\). For each such plane, we then eliminate the unit-circle variables and obtain a polynomial obstruction in one variable. Finally, we fix these obstruction polynomials over \(\F_{256}\) and use a self-reciprocal factor criterion to exclude their norm-one roots throughout the required infinite family of extensions.

Since the ordinary covering radius of the binary Zetterberg code is three, the two-summand criterion~\eqref{eq:~pairtrace} shows that its deep syndromes are
\begin{equation*}
 \mathcal D = \{a\in\F_{q^2}\setminus(U\cup\{0\}):~        \tr(\N(a)^{-1})=0\}.
\end{equation*}
For \(a\in\F_q^*\), this condition reduces to \(\tr(a^{-1})=0\), since \(\N(a)=a^2\) and the absolute trace is invariant under squaring. We call a subspace \(W\) a \emph{deep-syndrome subspace} if every nonzero element of \(W\) belongs to \(\mathcal D\).

{\em Reduction of a six-column cover.} 
Up to coordinate permutation, the binary \([6,3,3]\) code is
unique~\cite{LS2024}. We use the representative with generator matrix
\begin{equation}\label{eq:~G633}
 G=
 \begin{pmatrix}
 1&1&0&1&0&0\\
 1&0&1&0&1&0\\
 0&1&1&0&0&1
 \end{pmatrix}.
\end{equation}
Its four codewords of weight three have supports \(\{1,2,4\}\), \(\{1,3,5\}\), \(\{2,3,6\}\), and \(\{4,5,6\}\). Any two of these supports meet in exactly one coordinate, these six pairwise intersections are all distinct, and the four corresponding codewords sum to zero.

\begin{proposition}\label{prop:~six-reduction}
Let \(W\le\F_{q^2}\) be a three-dimensional deep-syndrome subspace.
If \(W\) can be covered by at most six columns of \(U\), then there
exists an affine plane \(T\subset W\), with \(0\notin T\), such that
for any three distinct elements \(a,b,c\in T\), there are
\(x,y,z\in U\) satisfying
\begin{equation}\label{eq:~sixcolumns}
 a+x+y,\qquad b+x+z,\qquad c+y+z\in U.
\end{equation}
There are exactly seven possible affine planes \(T\).
\end{proposition}

\begin{proof}
Suppose that \(W\) is covered by \(I=\{u_1,\ldots,u_\ell\}\subseteq U\), where \(\ell\le6\).
By Remark~\ref{rem:~short-code-lifting}, this common cover gives a binary linear code
\(D\le\F_2^\ell\) with parameters \([\ell,3,\ge3]_2\), together with an isomorphism
\[\sigma|_D:~D\longrightarrow W, \quad
 \sigma(e_1,\ldots,e_\ell)=\sum_{i=1}^{\ell}e_i u_i.\]
The Hamming bound~\cite[Theorem~1.12.1]{HuffmanPless2003} gives \(2^3(1+\ell)\le2^\ell\), which fails for \(\ell\le5\). Hence \(\ell=6\).
The minimum distance of \(D\) cannot be at least four, since the Griesmer bound~\cite[Theorem~2.7.4]{HuffmanPless2003} would then require \(6\ge4+2+1=7,\) a contradiction. Thus, \(D\) is a binary \([6,3,3]\) code. After permuting the six covering columns if necessary, we may therefore identify the coordinates of \(D\) with those of the code in~\eqref{eq:~G633}.

Let \(d_1,d_2,d_3,d_4\) be its four weight-three codewords, and put \(t_i=\sigma(d_i)\in W\). Since \(\sigma|_D\) is an isomorphism, the support relations among the \(d_i\)'s are transferred faithfully to the corresponding syndromes. In particular, \(d_1+d_2+d_3+d_4=0\), and hence \(t_1+t_2+t_3+t_4=0.\) Writing \(a=t_1\), \(b=t_2\), and \(c=t_3\), we obtain \(T=\{a,b,c,a+b+c\}\subset W\setminus\{0\}.\) Any three elements of \(T\) are linearly independent, because the corresponding three weight-three codewords are linearly independent.
Moreover, \(T=a+\spn{a+b,a+c},\) so \(T\) is a coset of a two-dimensional subspace of \(W\) and does not contain zero. Equivalently, \(T=\{w\in W:~\varphi(w)=1\}\) for some nonzero binary linear functional \(\varphi\) on \(W\).
Since \(\dim_{\F_2}W=3\), there are \(2^3-1=7\) nonzero linear functionals, and hence exactly seven affine planes in \(W\) that do not contain zero.

It remains to translate the support structure of the \([6,3,3]\) code back into column representations. Choose any three distinct elements of \(T\), and denote them by \(a,b,c\). Let \(d_a,d_b,d_c\) be the corresponding weight-three codewords. By the support pattern above, each pair of these three words has exactly one common coordinate, and the three common coordinates are distinct. Denote the associated columns of \(U\) by \(x,y,z\), where \(x\) is shared by \(d_a,d_b\), \(y\) by \(d_a,d_c\), and \(z\) by \(d_b,d_c\).
Each of the three supports has one remaining coordinate not shared with the other two. Let the corresponding columns be \(u_a,u_b,u_c\), respectively. Since \(\sigma(d_a)=a\), \(\sigma(d_b)=b\), and \(\sigma(d_c)=c\), and the field has characteristic two, we have
\[
 a=x+y+u_a,\quad
 b=x+z+u_b,\quad
 c=y+z+u_c.
\]
Therefore, 
\[
 u_a=a+x+y,\quad
 u_b=b+x+z,\quad
 u_c=c+y+z.
\]
All three \(u_a,u_b,u_c\) are among the six covering columns and hence belong to \(U\). This proves~\eqref{eq:~sixcolumns}.
\end{proof}

Thus, a six-column cover can exist only if one of the seven affine planes admits a solution of~\eqref{eq:~sixcolumns}. We next eliminate \(y\) and \(z\) and obtain a necessary polynomial condition on \(x\).

{\em Polynomial elimination.}
We next eliminate the unit-circle variables in \eqref{eq:~sixcolumns}. Throughout this argument, \(X\) denotes an indeterminate, while \(x\in U\) denotes a possible solution of \eqref{eq:~sixcolumns}. All auxiliary expressions are first defined in \(\F_q[X]\) and are evaluated at \(X=x\) only after a solution is assumed to exist. This distinction is important because the final obstruction will be a polynomial that can be studied uniformly over field extensions.

Let \(a,b,c\in\F_q\setminus\{0,1\}\). In \(\F_q[X]\), put
\begin{equation}\label{eq:~elim-first}
\begin{gathered}
 p_a=X+a,~~ p_b=X+b,~~
 S_a=1+aX,~~ S_b=1+bX,~~
 D=1+c(a+b+c),\\
 E_1=cXp_b+(p_a+c)S_b,~~
 F_1=cXp_a+(p_b+c)S_a,~~
 G_1=X(p_aS_b+p_bS_a).
\end{gathered}
\end{equation}
Next define
\begin{equation}\label{eq:~elim-second}
\begin{aligned}
 L_2&=E_1^2+p_bDE_1S_b+Xp_bD^2S_b,\\
 L_1&=p_b(E_1F_1+DG_1),\\
 L_0&=G_1^2+p_bG_1F_1S_b+Xp_bF_1^2S_b,\\
 L&=p_aL_2S_a+L_1S_b,\\
 M&=Xp_aL_2S_a+L_0.
\end{aligned}
\end{equation}
Finally, set
\begin{equation}\label{eq:~elim-R}
 R_{a,b,c}(X) =M^2+p_aS_aML+Xp_aS_aL^2.
\end{equation}
These polynomials arise naturally from two successive elimination steps:~ first \(z\) is eliminated, then \(y\).

\begin{lemma}\label{lem:~elimination}
Suppose that \(a,b,c\in\F_q\setminus\{0,1\}\) and \(x,y,z\in U\) satisfy~\eqref{eq:~sixcolumns}. Then
\[R_{a,b,c}(x)=0.\]
\end{lemma}

\begin{proof}
From now on, all expressions in \eqref{eq:~elim-first}--\eqref{eq:~elim-R} are evaluated at \(X=x\). Since \(\F_q\cap U=\{1\}\) and \(a,b\ne1\), the quantities \(S_a(x)=1+ax\), \(S_b(x)=1+bx\), \(x+a\), and \(x+b\) are all nonzero. For instance, \(S_a(x)=0\) would give \(x=a^{-1}\in\F_q\cap U=\{1\}\), and hence \(a=1\), a contradiction.
We first convert the first two unit-circle conditions in \eqref{eq:~sixcolumns} into quadratic equations. If \(u,v\in U\) and \(u+v=d\), then \(d^q=u^{-1}+v^{-1}=\frac{d}{uv},\) so \(uv=d/d^q\). Applying this to the pair \(y,a+x+y\), whose sum is \(a+x\), we have
\[y(a+x+y) =\frac{a+x}{a+x^{-1}} =\frac{x(x+a)}{ax+1}.\]
Similarly, we have \(z(b+x+y) =\frac{b+x}{b+x^{-1}} =\frac{x(x+b)}{bx+1}.\) Thus
\begin{equation}\label{eq:~elim-quadratics}
 y^2+p_ay+A=0~~{\rm and}~~ z^2+p_bz+B=0,
\end{equation}
where \(A=\frac{Xp_a}{S_a}\) and \(B=\frac{Xp_b}{S_b}\).
We now use the remaining condition \(c+y+z\in U\). Since \(c\in\F_q\) and \(y,z\in U\), we have \(c^q=c\), \(y^q=y^{-1}\), and \(z^q=z^{-1}\). Hence \(\N(c+y+z)=1\), and after multiplying by \(yz\) we obtain
\[ cyz(c+y+z)+yz+(y+z)(c+y+z)=0.\]
Using~\eqref{eq:~elim-quadratics} to replace \(y^2\) and \(z^2\), this relation reduces to
\begin{equation}\label{eq:~elim-bilinear}
 Dyz+Ey+Hz+G=0,
\end{equation}
where \(E=\frac{E_1}{S_b}\), \(H=\frac{F_1}{S_a}\), and \(G=\frac{G_1}{S_aS_b}\).

We next eliminate \(z\). Rewrite~\eqref{eq:~elim-bilinear} as \((Dy+H)z=Ey+G\). Rather than dividing by \(Dy+H\), which might vanish, multiply the
second equation in~\eqref{eq:~elim-quadratics} by \((Dy+H)^2\) and
use the preceding identity. This gives
\[(E^2+p_bDE+BD^2)y^2 +p_b(EH+DG)y +(G^2+p_bGH+BH^2)=0.\]
This identity remains valid even when \(Dy+H=0\). The three coefficients above are
\(\frac{L_2}{S_b^2}\), \(\frac{L_1}{S_aS_b}\), and \(\frac{L_0}{S_a^2S_b^2}\), respectively. Clearing the nonzero denominators yields
\[ S_a^2L_2y^2+S_aS_bL_1y+L_0=0.\]
Substituting \(y^2=p_ay+Xp_a/S_a\) from~\eqref{eq:~elim-quadratics}, this quadratic relation becomes \(S_aLy+M=0,\) where \(L\) and \(M\) are exactly those defined in
\eqref{eq:~elim-second}.

It remains to eliminate \(y\). Again, we avoid dividing by the possibly vanishing factor \(L(x)\). Multiplying the first equation in \eqref{eq:~elim-quadratics} by \((S_aL)^2\) and using \(S_aLy=M\), we obtain
\[ M^2+p_aS_aML+Xp_aS_aL^2=0.\]
By~\eqref{eq:~elim-R}, this is precisely \(R_{a,b,c}(x)=0\).
Thus, every solution of~\eqref{eq:~sixcolumns} gives a root of the one-variable polynomial \(R_{a,b,c}(X)\). No division by the potentially vanishing factors \(Dy+H\) or \(L(x)\) has been used, so the argument also covers all degenerate cases.
\end{proof}

The lemma reduces the original three-variable unit-circle system to a single polynomial obstruction in \(x\). To make this obstruction effective over an infinite family of extensions, we next choose the coefficients from a fixed finite subfield.

{\em Norm-one roots and extension degrees.}
We next recall the self-reciprocal factor criterion of Meyn and G\"otz~\cite[Theorem~1]{MeynGoetz1989}, which will allow us to control possible norm-one roots uniformly over field extensions.
Let \(q_*\) be a power of two, and let \(P\in\F_{q_*}[X]\) be a nonzero polynomial with \(P(1)\ne0\). Suppose that \(P(\xi)=0\) and that \(\xi^{q_*^r+1}=1\). Then
\begin{equation}\label{eq:~root-degree}
 [\F_{q_*}(\xi):~\F_{q_*}]=2e,~~
 e\mid r,~~
 \frac{r}{e}\ \text{is odd},~~
 2e\le\deg P.
\end{equation}
Moreover, whenever \(e\mid j\) and \(j/e\) is odd, the same element \(\xi\) also satisfies
\(\xi^{q_*^j+1}=1\).

In particular, if \(\deg P=12\), then every possible norm-one root of \(P\) has degree \(2e\) over \(\F_{q_*}\) with \(1\le e\le6\), regardless of the size of the ambient extension.
Thus, an a priori unbounded family of extension degrees can be reduced to finitely many possibilities for \(e\), and hence to a small number of fixed norm-one equations.

{\em Fixed witnesses over \(\F_{256}\).} 
We now choose the three-dimensional deep-syndrome subspaces used in the proof. The key point is that, when \(4\mid s\), these witnesses can be chosen inside a fixed subfield, independently of the extension degree. Write \(s=4r\). Then \(m=2s=8r\), and hence \(q=2^{8r}=256^r.\) 
Therefore,  \(\F_{256}\subseteq\F_q\subseteq\F_{q^2}.\)
Thus, every three-dimensional \(\F_2\)-subspace of \(\F_{256}\) is also a three-dimensional \(\F_2\)-subspace of \(\F_q\), and hence of the syndrome space \(\F_{q^2}\). Consequently, although the ambient field varies with \(r\), the candidate subspace \(W\) may be chosen once and for all inside the fixed field \(\F_{256}\).

We use two such witnesses, \(W_{\mathrm o},W_{\mathrm e}\le\F_{256}\), according to the parity
of \(r\). The parity determines how the absolute trace \(\Tr_{\F_q/\F_2}\) restricts to \(\F_{256}\), and hence which witness becomes a deep-syndrome subspace in \(\F_q\).
Since all affine planes of either witness are contained in \(\F_{256}\), the elements \(a,b,c\) entering the elimination also lie in \(\F_{256}\). It follows that the corresponding obstruction polynomials \(R_{a,b,c}(X)\), and later \(P_T(X)\), have coefficients in \(\F_{256}\) and are independent of \(r\). This is what allows the same finite collection of polynomials to control all extensions \(\F_{256^r}\).
Fix $q_*=256$ and a root $\theta$ of the irreducible polynomial $X^8+X^4+X^3+X^2+1$ over $\F_2$, so that
\begin{equation}\label{eq:~field256}
 \F_{256}=\F_2(\theta),~~ \theta^8+\theta^4+\theta^3+\theta^2+1=0.
\end{equation}
For \(j=\sum_{i=0}^7 j_i2^i\), write \([j]_\theta=\sum_{i=0}^7 j_i\theta^i.\)
Thus, the integers appearing below are labels for elements of \(\F_{256}\) in this fixed polynomial basis, rather than integer residues. Define two three-dimensional binary subspaces of \(\F_{256}\) by
\begin{equation*}
 W_{\mathrm o}=\spn{[2]_\theta,[9]_\theta,[69]_\theta}~~{\rm and}~~
 W_{\mathrm e}=\spn{[16]_\theta,[33]_\theta,[75]_\theta}.
\end{equation*}
Writing $W^*=W\setminus\{0\}$ and listing field elements by their integer labels, their nonzero elements are
\begin{equation*}
\begin{aligned}
 W_{\mathrm o}^{*}=\{2,9,11,69,71,76,78\}~~{\rm and}~~
 W_{\mathrm e}^{*}=\{16,33,49,75,91,106,122\}.
\end{aligned}
\end{equation*}
Both subspaces avoid \(1\).
For the representation~\eqref{eq:~field256}, it follows from Magma~\cite{magma} that the absolute trace on the polynomial basis \(1,\theta,\ldots,\theta^7\) gives
\[\Tr_{\F_{256}/\F_2}(\theta^i)=
\begin{cases}
1,&i=5,\\
0,&i\ne5.
\end{cases}\]
Hence the absolute trace of an element \(\sum_{i=0}^7 a_i\theta^i\) is simply its coefficient \(a_5\) of \(\theta^5\).
The inverses of the elements of \(W_{\mathrm o}^{*}\), in the order
listed above, have labels
\begin{equation}\label{eq:~inverse-certificate}
 142,~~157,~~152,~~78,~~4,~~22,~~69.
\end{equation}
Each of these elements has zero coefficient of \(\theta^5\). Hence all seven nonzero elements of \(W_{\mathrm o}\) have inverse trace zero over \(\F_{256}\).

For each of the seven affine planes \(T\) contained in either witness
and avoiding zero, we order the four labels increasingly and take the
first three as \(a,b,c\). Substituting these values into
\eqref{eq:~elim-first}--\eqref{eq:~elim-R} and carrying out the
arithmetic in \(\F_{256}[X]\), we obtain
\begin{equation}\label{eq:~normalized-obstruction}
 R_{a,b,c}(X)=\kappa_T X^2P_T(X),
\end{equation}
where \(\kappa_T\in\F_{256}^*\) and \(P_T\) is a monic reciprocal polynomial of degree twelve.
We write
\begin{equation}\label{eq:~palindrome-format}
 P_T(X)
 =\sum_{i=0}^{5}c_i(X^i+X^{12-i})+c_6X^6.
\end{equation}
By Magma~\cite{magma}, the coefficients of all fourteen polynomials are listed in Table~\ref{tab:~polynomial-certificate}, where the coefficients are interpreted as in \eqref{eq:~palindrome-format}. Since the characteristic is two, evaluating~\eqref{eq:~palindrome-format} at \(X=1\) gives \(P_T(1)=c_6\), which is nonzero in every row.

\begin{table}[htbp]
\centering
\caption{Complete normalized obstruction polynomials over \(\F_{256}\)}
\label{tab:~polynomial-certificate}
\small
\setlength{\tabcolsep}{10pt}
\renewcommand{\arraystretch}{1.15}
\begin{tabular}{clrrrrrrr}
\toprule
Witness & \(T\) & \(c_0\)&\(c_1\)&\(c_2\)&\(c_3\)&\(c_4\)&\(c_5\)&\(c_6\)\\
\midrule
\(W_{\mathrm o}\) & \(\{2,9,69,78\}\) &1&11&248&11&57&0&95\\
&\(\{2,9,71,76\}\) &1&11&196&159&74&202&234\\
&\(\{2,11,69,76\}\) &1&9&226&133&207&120&49\\
&\(\{2,11,71,78\}\) &1&9&180&155&107&153&83\\
&\(\{9,11,69,71\}\) &1&2&133&34&162&47&105\\
&\(\{9,11,76,78\}\) &1&2&74&190&40&88&232\\
&\(\{69,71,76,78\}\) &1&2&122&126&44&255&83\\
\midrule
\(W_{\mathrm e}\) &\(\{16,33,75,122\}\) &1&49&215&29&178&68&29\\
&\(\{16,33,91,106\}\) &1&49&236&190&225&3&10\\
&\(\{16,49,75,106\}\) &1&33&188&242&187&189&140\\
&\(\{16,49,91,122\}\) &1&33&52&58&251&12&220\\
&\(\{33,49,75,91\}\) &1&16&176&58&38&89&36\\
&\(\{33,49,106,122\}\) &1&16&238&122&51&51&103\\
&\(\{75,91,106,122\}\) &1&16&9&92&240&96&168\\
\bottomrule
\end{tabular}
\end{table}

The remaining certificates rule out the relevant norm-one roots:~
\begin{equation}\label{eq:~gcd-certificates}
\begin{array}{ll}
 \gcd(P_T,X^{q_*^{15}+1}+1)=1,
 & T\subset W_{\mathrm o}^{*},\\[2pt]
 \gcd(P_T,X^{q_*^{4}+1}+1) = \gcd(P_T,X^{q_*^{6}+1}+1)=1,
 & T\subset W_{\mathrm e}^{*}.
\end{array}
\end{equation}
Here each line refers only to the seven affine planes listed for the corresponding witness.

These coprimality relations are verified entirely inside \(\F_{256}[X]\); no large extension field needs to be constructed. For each relevant pair \((T,j)\), we first reduce \(X^{q_*^j+1}+1\) modulo \(P_T\), using repeated squaring, and denote the resulting remainder by \(H_{T,j}\). The Euclidean algorithm then produces polynomials \(A_{T,j},B_{T,j}\in\F_{256}[X]\) satisfying
$
 A_{T,j}P_T+B_{T,j}H_{T,j}=1.
$ 
Thus, \eqref{eq:~gcd-certificates} consists of exact algebraic certificates rather than numerical tests over finitely many extensions. The supplementary material records the remainders and the corresponding B\'ezout identities, and also verifies the factorizations~\eqref{eq:~normalized-obstruction}.

We can now use these fixed certificates to handle the whole infinite subfamily. It is enough to construct a three-dimensional subspace \(W\le\F_{q^2}\) satisfying \(\tau_U(W)>6\).

\begin{proof}[Proof of Theorem~\ref{thm:~binary-main}\textup{(ii)}]
Write \(s=4r\). Then \(m=8r\) and \(q=2^{8r}=q_*^r\), where \(q_*=256\). In particular, \(\F_{256}\subseteq\F_q\). The upper bound \(\rho_3(\mathcal Z_s)\le7\) follows from~\eqref{eq:~third-range}.
We first choose a three-dimensional deep-syndrome subspace according to the parity of \(r\). For \(a\in\F_{256}\), trace transitivity gives
\(\Tr_{\F_q/\F_2}(a) = (r\bmod2)\Tr_{\F_{256}/\F_2}(a).\)
If \(r\) is odd, we take \(W=W_{\mathrm o}\). By \eqref{eq:~inverse-certificate}, the inverses of all seven nonzero elements of \(W_{\mathrm o}\) have trace zero over \(\F_{256}\), and hence  also over \(\F_q\).
If \(r\) is even, we take \(W=W_{\mathrm e}\). In this case the displayed trace identity shows that every element of \(\F_{256}\) has absolute trace zero when regarded as an element of \(\F_q\). In particular, the inverse of every nonzero element of \(W_{\mathrm e}\) has trace zero over \(\F_q\).
Both witnesses avoid \(1\), while \(\F_q\cap U=\{1\}\). Hence in either case every nonzero element of \(W\) is a deep syndrome, and \(W\) is a three-dimensional deep-syndrome subspace of \(\F_q\).

Assume, for contradiction, that \(W\) can be covered by six columns. Then Proposition~\ref{prop:~six-reduction} yields an affine plane \(T\subset W\), three distinct elements \(a,b,c\in T\), and \(x,y,z\in U\) satisfying~\eqref{eq:~sixcolumns}.
Since \(W\subseteq\F_{256}\), the elements \(a,b,c\) all lie in \(\F_{256}\). Thus, the corresponding polynomial \(P_T(X)\) is one of the fixed degree-twelve polynomials in Table~\ref{tab:~polynomial-certificate}. By Lemma~\ref{lem:~elimination}, \(R_{a,b,c}(x)=0\). By \eqref{eq:~normalized-obstruction}, we have \(\kappa_Tx^2P_T(x)=0.\)
Since \(x\in U\), we have \(x\ne0\), and hence  \(P_T(x)=0\).
On the other hand, \(x\in U\) and \(q=q_*^r\), so \(x^{q_*^r+1}=1.\) Moreover, \(P_T(1)\ne0\). Therefore,  we may apply the root-degree criterion~\eqref{eq:~root-degree}. It follows that
\([\F_{q_*}(x):\F_{q_*}]=2e\)
for some integer \(e\) satisfying \(1\le e\le6\), \(e\mid r\), and \(r/e\) odd.

Assume first that \(r\) is odd. Then \(e\) is odd, so \(e\in\{1,3,5\}\). Each of these values divides \(15\), with \(15/e\) odd. Therefore,  the same root-degree criterion implies \(x^{q_*^{15}+1}=1\). Thus, \(x\) would be a common root of \(P_T(X)\) and \(X^{q_*^{15}+1}+1\), contradicting the first line of~\eqref{eq:~gcd-certificates}.

Now suppose that $r$ is even. Write $\nu_2(n)$ for the exponent of $2$ in a positive integer $n$. Since \(r/e\) is odd, \(e\) and \(r\) have the same \(2\)-adic valuation:~ \(\nu_2(e)=\nu_2(r).\)
If \(\nu_2(r)=1\), then \(e\le6\) forces \(e\in\{2,6\}\). In either case \(e\mid6\) and \(6/e\) is odd, so \(x^{q_*^6+1}=1\), contradicting \(\gcd(P_T,X^{q_*^6+1}+1)=1\).
If \(\nu_2(r)=2\), then the only possibility with \(e\le6\) is \(e=4\). Hence \(x^{q_*^4+1}=1\), contradicting \(\gcd(P_T,X^{q_*^4+1}+1)=1\).
Finally, if \(\nu_2(r)\ge3\), then \(\nu_2(e)=\nu_2(r)\) implies \(e\ge8\), contradicting \(e\le6\).
Thus, no six-column cover of \(W\) exists. Hence \(\tau_U(W)\ge7\), and hence  \(\rho_3(\mathcal Z_s)\ge7\). Together with the previously established upper bound, this gives
\(\rho_3(\mathcal Z_s)=7\) whenever \(4\mid s.\)
Equivalently, \(\rho_3(\mathcal B_m)=7\) whenever \(8\mid m.\)
\end{proof}

\begin{remark}
Theorem~\ref{thm:~binary-main}\textup{(ii)} is computer-assisted only in its finite certificate component. Exact arithmetic over \(\F_{256}\) is used to verify the factorizations \eqref{eq:~normalized-obstruction}, the coefficients listed in Table~\ref{tab:~polynomial-certificate}, and the B\'ezout identities underlying~\eqref{eq:~gcd-certificates}.
The passage from these finite certificates to the infinite family is entirely theoretical. It follows from the self-reciprocal root-degree criterion~\eqref{eq:~root-degree} and the argument above, rather than from checking finitely many values of \(s\).
The direct certificate for \(s=4\) is retained in the supplementary material as an independent verification. We make no claim that \(\rho_3(\mathcal Z_s)=7\) for \(s\ge4\) when \(4\nmid s\).
\end{remark}

\section{Conclusion}\label{sec:~conclusion}

In this paper, we studied the generalized covering radii of generalized Zetterberg codes through their norm-one column configurations.
We first established general upper and lower bounds over arbitrary finite fields. The upper bound is obtained from a common-summand argument with an adapted basis, while the lower bounds come from projective counting and subfield constructions. In a broad range of parameters, these estimates restrict the generalized covering radius to two consecutive values.
We then investigated the extremal case \(\rho_t=t\). By estimating the relative intersections \(U\cap(W\setminus L)\) using additive Fourier coefficients and Kloosterman-sum bounds, we obtained explicit high-order ranges in which every \(t\)-dimensional syndrome subspace is internally spanned by the norm-one columns it contains.
Finally, for binary Zetterberg codes, we determined the second generalized covering radius in every extension degree and proved that the third generalized covering radius is seven for the infinite subfamily with \(8\mid m\).
These results show that the arithmetic structure of norm-one configurations remains effective in the higher-order covering setting and leads to both general bounds and exact generalized covering radii.

\paragraph{Acknowledgments.} ChatGPT Pro was used for improving the presentation of the paper. This research of Shitao Li was supported by the National Natural Science Foundation of China under Grant Nos. 12526612 and 126011020. The research of Yang Li was supported by the Nanyang Technological University Research under Grant 04INS000047C230GRT01. The research of Gaojun Luo was supported by the National Natural Science Foundation of China under Grant 12401690. 
The research of Zhonghua Sun was supported by the National Natural Science Foundation of China under Grant 12571573.

\end{document}